\documentclass[numbers,11pt]{elsarticle}
\usepackage[a4paper]{geometry}
\usepackage{amsmath}
\usepackage{amstext}
\usepackage{amsfonts}
\usepackage{amssymb}
\usepackage{amsthm}
\usepackage{epsfig}
\usepackage{bbm}
\usepackage[english]{babel}
\usepackage{natbib}
\usepackage{color}
\usepackage[T1]{fontenc}
\usepackage{mathrsfs, mathtools}
\mathtoolsset{showonlyrefs}

 \theoremstyle{definition}

 \theoremstyle{plain}
 \newtheorem{thm}{Theorem}[section]
 \newtheorem{prop}{Proposition}[section]
 \newtheorem{cor}{Corollary}[section]
 \newtheorem{lem}{Lemma}[section]
\newtheorem{Algorithm}{Algorithm}[section] 

 \theoremstyle{remark}

\setcitestyle{authoryear,round}

\usepackage{float}
\usepackage{graphicx}  
\usepackage[all]{xy}
\usepackage{amscd}
\usepackage{verbatim}
\usepackage{enumerate}

\usepackage{times}
\usepackage{amstext}
\usepackage{amsfonts}
\usepackage{amssymb}
 
\usepackage{subfig}
\usepackage{nextpage}
\usepackage{appendix}

\begin{document}
\begin{frontmatter}
\title{Affine Pricing and Hedging of Collateralized Debt Obligations}
\author[label1]{Zehra Eksi}
\ead{zehra.eksi@wu.ac.at}
\author[label2]{Damir Filipovi\'{c} \corref{cor1}}
\ead{damir.filipovic@epfl.ch}
\cortext[cor1]{Corresponding author}
\address[label1]{Institute for Statistics and Mathematics, WU-University of Economics and Business,
Welthandelsplatz 1, 1020 Vienna, Austria}
\address[label2]{Swiss Finance Institute and Ecole Polytechnique F\'{e}d\'{e}rale de Lausanne, Switzerland}

\begin{abstract}
This study deals with the pricing and hedging of single-tranche collateralized debt obligations (STCDOs). We specify an affine two-factor model in which a catastrophic risk component is incorporated. Apart from being analytically tractable, this model has the feature that it  captures the dynamics of super-senior tranches, thanks to the catastrophic component. We estimate the factor model based on the iTraxx Europe data with six tranches and four different maturities, using a quasi-maximum likelihood (QML) approach in conjunction with the Kalman filter. We derive the model-based variance-minimizing strategy for the hedging of STCDOs with a dynamically rebalanced portfolio on the underlying swap index. We analyze the actual performance of the variance-minimizing hedge on the iTraxx Europe data. In order to assess the hedging performance further, we run a simulation analysis where normal and extreme loss scenarios are generated via the method of importance sampling. Both in-sample hedging and simulation analysis suggest that the variance-minimizing strategy is most effective for mezzanine tranches in terms of yielding less riskier hedging portfolios and it fails to provide adequate hedge performance regarding equity tranches.
\end{abstract}

\end{frontmatter}
\section{Introduction}
\label{introduction}
 A collateralized debt obligation (CDO) is a structured product that is backed by a portfolio of credit risky assets. A synthetic CDO is a special type of CDO in which the underlying credit risky portfolio consists of single-name CDSs. Single tranche CDOs (STCDO), on the other hand, make it possible to take an exposure on a specific segment of the underlying portfolio. A CDO long position holder is exposed to two types of risk. The \emph{default risk} arises from the possibility of a default of an obliger and the \emph{market} or \emph{spread} risk is associated with the changes in the credit qualities and the interest rates. Thus, a sound model for the pricing and hedging of CDOs is expected to take both default and credit spread dynamics into account.

 Overall, one can categorize the portfolio credit derivative models according to two approaches. Under the \emph{bottom-up approach} the fundamental objects to be modeled are the loss processes of the portfolio constituents whose sum gives the total portfolio loss. In contrast, the \emph{top-down approach} aims to model the evolution of the aggregate portfolio loss process directly. For a systematic comparison of the top-down and bottom-up approaches we refer to  \cite[Section 2]{bcj} and \citet{gies}. One may also classify the portfolio credit models as \emph{static} or \emph{dynamic}. In static models the particular interest is on the default time distributions of constituents at a given point in time. This results in a model which does not allow for the consistent pricing across different maturities. On the other hand, the dynamic models specify the evolution of default time distribution or the total loss process depending on the top-down or bottom up framework that is followed. Prior to the financial crisis of 2008, copula-based models (see, e.g., \citep{Li}), which are static and focus on the default risk only, were the market standard due to the ease of implementation. The financial crisis 2008 showed that these models were inadequate. It highlighted the necessity for dynamic pricing and hedging models.

In this paper, we propose an affine two-factor model within the framework of \citet{FOS} to dynamically price and hedge STCDOs.  The proposed model is parsimonious and capable of capturing the dynamics of senior tranches, thanks to the \emph{catastrophic risk component} incorporated into the model. Moreover, the resulting loss given default distribution is a weighted mixture with stochastic weights: the first affine factor tunes the truncated exponential distribution, while the second factor drives the catastrophic component. 

We focus on the dynamic hedging of STCDOs with the underlying index swap by computing the variance-minimizing hedging strategy corresponding to our model specification. Then, we test the performance of the variance-minimizing hedge on the iTraxx Europe data. This analysis is complementary to \citet{FS} in the sense that it carries the theoretical framework of that study into implementation. The novelty of this analysis is also due to the data set we use. Although this data set does not include any default event, it covers a  sufficiently long time horizon which witnessed extreme market conditions such as the credit crisis in 2008.  
In order to investigate further on the performance of the hedging strategy, we run a simulation analysis. Specifically, along with the normal scenarios we generate extreme scenarios by employing the importance sampling technique. Both in-sample hedging and simulation analysis suggest that the variance-minimizing strategy is most effective for mezzanine tranches in terms of yielding less riskier hedging portfolios and it fails to provide adequate hedge performance regarding  equity tranches, which is in line with the results by \citet{ascheberg2013hedging}.

We estimate the factor model based on the iTraxx Europe data with six tranches and four different maturities, using a quasi-maximum likelihood (QML) approach in conjunction with the Kalman filter. The QML approach necessitates the knowledge of the first two conditional moments of the factor process. Accordingly, we use the polynomial preserving property of affine processes and compute the conditional mean and variance of the factor process explicitly. Our findings suggest that apart from being analytically tractable, the two-factor affine model with a catastrophic component is capable of describing the dynamics of the tranche data simultaneously. However, a two-factor affine model with the restriction of a zero catastrophic component is not able to capture the dynamics of senior tranches, in particular during the crisis period (for a similar result in the context of a structural model, see \citep{col12}).

To our knowledge,  \citet{DuGu}, in which correlated intensities are constructed for constituent names by using affine factor processes, is the first study addressing the dynamic framework for pricing CDOs. \citet{Sch}, \citet{SAP}, \citet{bslm}, \citet{fb}, \citet{FOS} and \citet{contm} are other examples of dynamic models for CDO pricing. \citet{Sch}, \citet{SAP} and \citet{FOS} are very much in the same spirit that they aim to model the evolution of the forward distribution of the loss process. This allows for a consistent incorporation of the dynamics of credit spreads in the modeling of  multi-name credit derivatives.
Inspired by the HJM framework ( see, \citep{hjm}) for a default-free term structure, \citet{FOS} develop a dynamic no-arbitrage setting for the evolution of the forward credit spreads.  Allowing for feedback and contagion effects, this framework provides a generalization for the aforementioned top-down models. Furthermore, under this general framework an analytically tractable class of doubly-stochastic affine term structure models is proposed.  

\citet{fb} study the dynamic hedging of STCDOs in a framework where spread risk is incorporated along with the default contagion (for a static \emph{infectious default model}, see, \citep{davis2001infectious}. Reckoning with the incompleteness of the market arising from the presence of spread and default risk, this study utilizes a variance-minimizing strategy for the hedging of STCDOs with underlying CDSs (see, e.g., \citep{SC2} for the review of variance-minimizing and other quadratic hedging approaches). Notably, \citet{fb} show that the variance minimization provides a model-based endogenous interpolation between the hedging against spread risk and default risk. For the hedging of a STCDO position with a dynamically rebalanced portfolio on the index, \citet{FS} derive the variance-minimizing hedging strategy based on a general top-down framework. \citet{ascheberg2013hedging} tests the performance of more than 10 credit models during the crisis period and find that the hedging performances of all models are unsatisfactory. In particular they showed that the top-down models fail to hedge equity tranches.

The remainder of this paper is structured as follows: Section 2 gives the basics of STCDOs, introduces the two-factor affine model as well as the corresponding variance minimizing strategy. Section 3 provides the estimation methodology in detail. Section 4, which is on the data implementation, presents the data set, discusses the practical issues and gives results on estimation, in-sample hedging and simulation. Section 5 summarizes the main findings and concludes the paper. The technical appendix contains auxiliary results and proofs. 

\section{Modeling Framework}
We fix a stochastic basis $(\Omega,\mathcal{F},(\mathcal{F}_t),\mathbb{P})$ satisfying the usual conditions and where $\mathbb{P}$ denotes the historical probability measure. We consider a CDO pool of credits with the overall outstanding notional normalized to $1$. The aggregate loss process, representing the ratio of CDO losses realized by time $t$, is indicated by $L_t$. It is assumed that
\[L_t=\sum_{s\leq t}\Delta L_s\]
is $[0,1]$-valued, non-decreasing marked-point process with absolutely continuous $\mathbb{P}$-compensator $\nu(t,dx)dt$.

We define, for any $T>0$ and $x\in [0,1]$, the hypothetical $(T,x)$-bond which pays $1_{\{L_T\leq x\}}$ at maturity $T$. That is, we have a defaultable zero-recovery zero-coupon bond, which pays $1$ if the realized loss fraction at $T$ is less than or equal to $x$, and $1$ otherwise. The $(T,x)$-bond price at time $t\leq T$ is denoted by $P(t,T,x)$. It follows that the risk-free zero-coupon bond price $P(t,T)$ equals $P(t,T,1)$. Throughout, we assume that the risk-free rate $r$ is constant, so that $P(t,T)=e^{-r(T-t)}$.

$(T,x)$-bonds have the {\it spanning property}: any European contingent claim $F(L_T)$ with absolutely continuous payoff function $F$ can be decomposed into the sum of $(T,x)$-bond payoffs, $F(L_T)=F(1)-\int_0^1 F'(x)1_{\{L_T\le x\}}dx$. Hence the claim can be replicated by the static portfolio with the value process
\begin{equation}\label{repli}
 F(1)P(t,T)-\int_0^1 F'(x)P(t,T,x)dx.
\end{equation}
This, in particular, allows for the pricing of a STCDO via $(T,x)$-bonds.
\subsection{Single-Tranche CDOs}
We recall that (see e.g.\cite[Chapter 9]{EFM}) a STCDO issued at time $0$ is specified by a sequence of coupon payment dates $0<T_1<\dots<T_n$, a tranche with attachment and detachment point $x_1<x_2$ in $[0,1]$, and a coupon rate $\kappa_0^{(x_1,x_2]}$. The attachment point $x_1$ indicates the level at which losses in the underlying CDO pool begin to erode the notional of the tranche. At the detachment point $x_2$ the full tranche is written down. Note that $(x_1,x_2)=(0,1)$ corresponds to the entire {\emph{index}}. The holder of a long position in a STCDO
\begin{itemize}
\item receives $\kappa_0^{(x_1,x_2]} \times H^{(x1,x2]}(L_{T_i})$ at $T_i$, $i=1,2,\dots,n$ (\textit{coupon leg})
\item pays $-\Delta H^{(x1,x2]}(L_t)=H^{(x1,x2]}(L_{t-})-H^{(x1,x2]}(L_t)$ at any time $t\le T_n$ where $\Delta L_t\neq0$ (\textit{protection leg})
\end{itemize}
where we define $H^{(x1,x2]}(x):=\int_{x_1}^{x_2}1_{\{x\leq y\}}dy=(x_2-x)^+-(x_1-x)^+$. It follows from \eqref{repli} that the value at time $t\le T_n$ of the coupon leg is given by $\kappa_0^{(x_1,x_2]}\times S^{(x1,x2]}(t)$, where we define the annuity factor
\begin{align}S^{(x1,x2]}(t)=  \sum_{t<T_i }\int_{x_1}^{x_2}P(t,T_i,x) dx,\end{align}
and the time $t$ value of the protection leg is
\begin{align}V^{(x_1,x_2]}_P(t)=\displaystyle\int_{x_1}^{x_2}\left(  1_{\{L_t\leq x\}}-P(t,T_n,x)-r\int_{t}^{T_n}P(t,u,x)du  \right)dx.\end{align}
The par-coupon rate at $t$ is then defined as the rate $\kappa_t^{(x_1,x_2]}$ by which $\kappa_0^{(x_1,x_2]}$ would need to be replaced for rendering the two legs equal in value. That is,
\begin{equation}\label{kappa}
\kappa_t^{(x_1,x_2]}=\displaystyle\frac{V^{(x_1,x_2]}_P(t)}{S^{(x1,x2]}(t)}.
\end{equation}
In turn, this implies that the time $t$ spot value, $\Gamma_t^{(x_1,x_2]}$, of a long position in this STCDO equals
\begin{equation}\label{gamma}
\Gamma_t^{(x_1,x_2]}=\left(\kappa_0^{(x_1,x_2]}-\kappa_t^{(x_1,x_2]}\right)D^{(x1,x2]}(t).
\end{equation}
The discounted gains process, $G_t^{(x_1,x_2]}$, from holding a long position in the STCDO equals the sum of the accumulated discounted cash flows,
\[ A_t^{(x_1,x_2]}=\kappa_0^{(x_1,x_2]} \displaystyle{\sum}_{T_i\leq t}e^{-rT_i}H(L_{T_i})+\int_0^te^{-ru}dH(L_u),\]
and the discounted spot value. That is,
\begin{equation}\nonumber
G_t^{(x_1,x_2]}=A_t^{(x_1,x_2]}+e^{-rt}\Gamma_t^{(x_1,x_2]}.
\end{equation}
One can show that the discounted gains process satisfies the following dynamics
 \begin{align}
dG^{(x_1,x_2]}_t&=\!\int_{x_1}^{x_2}\!\!\Big\{e^{-rt}\!\Big(\kappa^{(x_1,x_2]}_0\!\! \sum_{t<T_i}dP(t,T_i,x)\!-\!r P(t,T_i,x)dt\!+\! dP(t,T_n,x)-rP(t,T_n,x)dt\\
&+r1_{\{L_t\leq x\}}dt \!\Big)\!+\! d\gamma(t,x)\! \Big\}dx,
\label{G-dyn}\end{align}
where $\gamma(t,x)=re^{-rt}\int_t^{T_n}P(t,u,x)du$.

\subsection{Model Specification}
In the following, we specify an affine factor model for the stochastic evolution of the $(T,x)$-bond prices. First, we consider an ${\mathbb R}^2_+$-valued affine state process $(Y,Z)$ given by
\begin{eqnarray}\label{Y_P}
dY_t&=&\kappa_y\left(Z_t-Y_t\right)dt+\sigma_y\sqrt{Y_t}dW^y_t \\
\label{Z_P}dZ_t&=&\kappa_z\left(\theta_z-Z_t\right)dt+\sigma_z\sqrt{Z_t}dW^z_t
\end{eqnarray}
for parameters $\kappa_y \geq 0$, $\kappa_z\theta_z\geq 0$, $\sigma_y,\sigma_z\ge 0$, and where $W=(W^y,W^z)^\top$ is a two-dimensional $\mathbb{P}$-Brownian motion. The factor $Z$ represents the stochastic mean reversion level of factor $Y$. Note that, during the empirical analysis given in Section~\ref{daI}, we also consider the nested one-factor model with $Y$ having the dynamics given above and with a constant mean reversion level $Z_t\equiv\theta_y\ge 0$.

In order to obtain arbitrage-free prices, we need the dynamics of the state process under a risk-neutral pricing measure ${\mathbb Q}\sim{\mathbb P}$. To preserve the affine structure of the state process under ${\mathbb Q}$ (see e.g.\ \citet{CFK} for details) we specify the market price of risk  process $\lambda_t=(\lambda^y_t,\lambda^z_t)$ in the following way
\begin{equation}\nonumber
\lambda^y_t=\frac{\lambda_y\sqrt{Y_t}}{\sigma_y},\quad \lambda^z_t=\frac{\lambda_z\sqrt{Z_t}}{\sigma_z}.
\end{equation}
Then $\widetilde{W}_t=(\widetilde{W}^y_t,\widetilde{W}^z_t)=W_t+\int_0^t \lambda_s^\top ds$ becomes a Brownian motion under $\mathbb{Q}$ with Radon--Nikodym density process
\[ \frac{d\mathbb Q}{d\mathbb P}|_{\mathcal F_t} = \exp\left(-\int_0^t \lambda_s dW_s-\frac{1}{2}\int_0^t \|\lambda_s\|^2ds\right) .\]
The dynamics of the state process under $\mathbb{Q}$ read
\begin{eqnarray}\nonumber
dY_t&=&\left(\kappa_y+\lambda_y\right)\left(\frac{\kappa_y}{\kappa_y+\lambda_y}Z_t-Y_t\right)dt+\sigma_y \sqrt{Y_t}d\widetilde{W}^y_t\\
\nonumber dZ_t&=&\left(\kappa_z+\lambda_z\right)\left(\frac{\kappa_z}{\kappa_z+\lambda_z}\theta_z-Z_t\right)dt+\sigma_z \sqrt{Z_t}d\widetilde{W}^z_t.
\end{eqnarray}
We assume no market price of default event risk. That means the ${\mathbb Q}$-compensator $\nu^{\mathbb Q}(t,dx)$ of the loss process $L$ is equal to the ${\mathbb P}$-compensator:
\begin{equation}\nonumber
  \nu^{\mathbb Q}(t,dx)=\nu(t,dx).
\end{equation}
The following theorem provides the affine specification for the $(T,x)$-bond market.
\begin{thm}\label{thmaffine} 
Let $\alpha$, $\beta_{y}$ and $\beta_{z}$ be some non-increasing and c\`{a}dl\`{a}g functions with $\alpha(x)= r$ and $\beta_{y}(x)=\beta_{z}(x)=0$ for $x\geq 1$. Then, under the above assumptions, there exists a loss process $L$ which is unique in law
such that
\begin{equation}\label{p2}P(t,T,x)=1_{\{L_t\leq x\}}e^{-A(T-t,x)-B_y(T-t,x)Y_t-B_z(T-t,x)Z_t}\end{equation}defines an arbitrage-free $(T,x)$-bond market,\footnote{That is, $e^{-rt}P(t,T,x)$, $0\le t\le T$, is a ${\mathbb Q}$-martingale, for any $T>0$ and $x\in [0,1]$.}where the functions $A$, $B_y$ and $B_z$ solve the Riccati equations
\begin{align}
\nonumber \partial_{\tau} A(\tau,x)&=\alpha(x)+\kappa_z\theta_zB_z(\tau,x) ,\\
\nonumber A(0,x)&=0, \\
\nonumber \partial_{\tau}B_y(\tau,x)&=\beta_y(x)-(\kappa_y+\lambda_y)B_y(\tau,x)-\frac{1}{2}\sigma_y^2B_y(\tau,x)^2  ,\\
\nonumber B_y(0,x)&= 0, \\
\nonumber \partial_{\tau} B_z(\tau,x)&=\beta_z(x)+\kappa_yB_y(\tau,x)-(\kappa_z+\lambda_z)B_z(\tau,x)-\frac{1}{2}\sigma_z^2B_z(\tau,x)^2,\\
\nonumber B_z(0,x)&=0.
\end{align}
Moreover, the compensator of $L$ is given by
\begin{equation}\nonumber
\nu(t,(0,x])=\alpha(L_{t-})-\alpha(L_{t-} +x)+(\beta_y(L_{t-})-\beta_y(L_{t-} +x))Y_t+(\beta_z(L_{t-})-\beta_z(L_{t-} +x))Z_t.
\end{equation}
\end{thm}
\begin{proof}
Proof follows the same arguments as in Theorem 7.2 of \citet{FOS}.
\end{proof}
Plugging \eqref{p2} into \eqref{G-dyn} yields the following corollary (see also \cite[Eqn.~(20)]{FS}).
\begin{cor}\label{coraffine}
The implied discounted gains process is a square-integrable ${\mathbb Q}$-martingale with dynamics 
\begin{equation}\nonumber
dG_t^{(x_1,x_2]}=e^{-\int_0^t r_udu}\left(  B_t^{(x_1,x_2]}d\widetilde{W}_t+\int_{(0,1]}C_t^{(x_1,x_2]}(\xi)(\mu(dt,d\xi)-\nu(t,d\xi)dt) \right)
\end{equation}
where $\mu(dt,dx)$ denotes the integer-valued random measure associated to the jumps of the loss process $L$, and
\begin{align*}
B_t^{(x_1,x_2]}&=\int_{(x_1,x_2]}\Bigg{\{}\kappa_0^{(x_1,x_2]}\sum_{t<T_i}P(t,T_i,x)\beta(t,T_i,x)\\
&+P(t,T_n,x)\beta(t,T_n,x)+r\int_t^{T_n}P(t,u,x)\beta(t,u,x)du\Bigg{\}}dx\\
C_t^{(x_1,x_2]}(\xi)&=-\int_{(x_1,x_2]}\Bigg{\{}\kappa_0^{(x_1,x_2]}\sum_{t<T_i}P(t-,T_i,x)1_{\{L_{t-}+\xi>x\}}\\
&+P(t-,T_n,x)1_{\{L_{t-}+\xi>x\}}+r\int_t^{T_n}P(t-,u,x)1_{\{L_{t-}+\xi>x\}}du\Bigg{\}}dx
\end{align*}
with
\begin{eqnarray}\nonumber
\beta(t,T,x)&=&\left(-B_y(T-t,x)\sigma_y \sqrt{Y_t}\,,\,-B_z(T-t,x)\sigma_z \sqrt{Z_t}\right)^{\top}.
\end{eqnarray}
\end{cor}

Here, we emphasize that the functions $\alpha$, $\beta_y$ and $\beta_z$ in Theorem~\ref{thmaffine} are exogenous and determine the default event arrival intensity
\begin{equation}\label{arrintense}\Lambda_t =\nu(t, (0, 1])=\alpha(L_{t-})-r+\beta_y(L_{t-})Y_t+\beta_z(L_{t-})Z_t.\end{equation}
as well as the cumulative loss given default distribution function
\[F_L(t, x) = \frac{\nu(t, (0, x])}{\Lambda_t  }.\]

We henceforth fix some nonnegative parameters $\gamma, a_0, b_0, c_0\geq 0$ , and set
\begin{gather*}
 \alpha(x)= \gamma\left(e^{-a_0(x\wedge 1)}-e^{-a_0}\right)+r,\\
\beta_y(x)= e^{-b_0(x\wedge 1)}-e^{-b_0},\quad
\beta_z(x)=c_0 1_{[0,1)}(x).
\end{gather*}
This yields
\begin{align*}
 \nu\left(t,(0,x]\right)&=\!\gamma\!\left(e^{\!-a_0(L_{t-}\wedge 1)}\!-\!e^{\!-a_0(L_{t-}\!+\!x\wedge 1)}\!\right)\!\!+\!\!\left(e^{\!-b_0(L_{t-}\wedge 1)}\!-\!e^{\!-b_0(L_{t-}+x\wedge 1)}\!\right)\! Y_t\!+\!c_0 1_{\{1-L_{t-} \leq x\}}Z_t\\
 \Lambda_t&=\!\gamma\!\left(e^{\!-a_0(L_{t-}\wedge 1)}\!-\!e^{\!-a_0}\!\right)\!+\!\left(e^{\!-b_0(L_{t-}\wedge 1)}\!-\!e^{\!-b_0}\!\right)Y_t\!+\!c_01_{[0,1)}(L_{t-})Z_t,
\end{align*}
and hence the loss given default tail distribution becomes
\[ 1-G_L(t,x) = \frac{\gamma\!\left(e^{\!-a_0((L_{t-}+x)\wedge 1)}\!-\!e^{\!-a_0}\!\right)\!+\!\left(e^{\!-b_0((L_{t-}+x)\wedge 1)}\!
-\! e^{\!-b_0}\!\right)\!Y_t\!+\!c_01_{[0,1)}(L_{t-}+x)Z_t}{\gamma\!\left(\!e^{\!-a_0L_{t-}}\!-\!e^{\!-a_0}\!\right)\!
+\!\left(\!e^{\!-b_0L_{t-}}\!-\!e^{\!-b_0}\!\right)\!Y_t\!+\!c_01_{[0,1)}(L_{t-})Z_t}.\]
Evidently, this is a weighted mixture of truncated exponential distributions and a point mass at $1-L_{t-}$. The latter models a catastrophic default event which extinguishes the entire CDO pool. This {\emph{ catastrophic risk component}} will be crucial for fitting the super-senior tranche spread levels, see also \citet{CDG} and \citet{col12}. Notice that the loss given default distribution is not static. The mixture weights are stochastic: $Y_t$ tunes the truncated exponential distribution part, while $Z_t$ drives the catastrophic component.

Finally, integration by parts gives the expected loss given default as
\begin{align}\label{ELG}
 & \int_0^1 x\, G_L(t,dx)=\! \int_0^1 (1-G_L(t,x))dx\\
  \nonumber &=\frac{\!\frac{\gamma}{a_0}\!\left(\!e^{\!-a_0L_{t-}}\!-\!e^{\!-a_0}\!(a_0(1\!-\!L_{t-})\!+\!1)\!\right)\!
+\!\frac{1}{b_0}\!\left(\!e^{\!-b_0L_{t-}}\!-\! e^{\!-b_0}(b_0(1\!-\!L_{t-})\!+\!1)\!\right)\!Y_t\!+\!c_0(1\!-\!L_{t-})Z_t}{\gamma\!
\left(\!e^{\!-a_0L_{t-}}\!-\!e^{\!-a_0}\!\right)\!+\!\left(\!e^{\!-b_0L_{t-}}\!-\!e^{\!-b_0}\!\right)\!Y_t\!+\!c_01_{[0,1)}(L_{t-})Z_t}.
\end{align}

\subsection{Variance-minimizing Hedge}
Recall that the index corresponds to some weighted portfolio of single name CDSs. It is understood that the index can be replicated by holding the respective positions in the constituent CDSs. We thus consider the problem of hedging a long position in a STCDO with a dynamically rebalanced self-financing portfolio in the entire index and the money market account. Due to the presence of both spread risk and the default risk in the model, the market formed by the index and the money market account is incomplete. In particular, we cannot perfectly replicate a long position in a STCDO with trading in the index and money market alone. Instead, we will consider \emph{variance-minimizing strategies}.

A variance-minimizing hedge minimizes the $\mathbb{Q}$-conditional variance of the hedging error. This method yields a self-financing hedging strategy, which coincides with the perfect replication in a complete market situation. We recall from Corollary~\ref{coraffine} that the discounted gains process $G^{(x_1,x_2]}$ is a square-integrable ${\mathbb Q}$-martingale. It follows that, for any time interval $[0,T]$ with $T\le T_n$, the self-financing strategy $\phi=\phi^{VM}$, with
\[ \phi^{VM}_t   = -\frac{d\langle G^{(x_1,x_2]},G^{(0,1]}\rangle_t}{d\langle
G^{(0,1]} \rangle_t}=-\frac{B_t^{(x_1,x_2]}B_t^{(0,1]}+\int_{(0,1]}C_t^{(x_1,x_2]}(\xi)C_t^{(0,1]}(\xi)\nu(t,d\xi)}{(B_t^{(0,1]})^2+\int_{(0,1]}(C_t{(0,1]}(\xi))^2\nu(t,d\xi)},\] along with the initial capital
$c=G_0^{(x_1,x_2]}=0$ is the unique minimizer of the quadratic
hedging error
\[ \inf_{ c,\phi} \mathbb{E^{Q}}\left[ \left(G_T^{(x_1,x_2]}- c+\int_0^{T}
\phi_t\,dG_t^{(0,1]}\right)^2 \right].\]
Here the infimum is taken over all $c\in {\mathbb R}$ and
predictable processes $\phi$ with
\[ \mathbb{E^{Q}}\left[\int_0^{T_n} \phi_t^2\,d\langle
G^{(0,1]}\rangle_t\right]< \infty.\]
See e.g.\  \cite[Theorem 5.1]{FS}. The strategy $\phi^{VM}$ is called the variance-minimizing strategy. Note that it does not depend on the hedging time horizon $T$.

\section{Parameter Estimation and Filtering}\label{EM}

In the current framework the fundamental object modeled is the hypothetical term-structure of $(T,x)$-bonds, which is not directly observable in the market. However, given the market observable par coupon rates for all tranches and the index, the term-structure of $(T,x)$-bonds can be estimated via inverting the formula \eqref{kappa}. More precisely, we assume there are $J$ tranches with attachment/detachment points $0=x_0<x_1<\dots<x_J=1$.  Denoting time to maturity with $\tau=T-t$, we first estimate the zero-coupon discount curve
\begin{equation}\label{dcurve}
\tau \mapsto D(t,\tau,j)=\frac{1}{x_{j}-x_{j-1}}\int_{x_j-1}^{x_{j}}P(t,t+\tau,x)dx
\end{equation}
for all tranches $(x_{j-1},x_{j}]$. This, in turn, gives the implied zero-coupon spread curve
\begin{equation}\label{rcurve}
R(t,\tau,j)=-\frac{1}{\tau}\log D(t,\tau,j)-r.
\end{equation}
Finally, one can get the term structure of $(T,x)$-bonds via interpolating \eqref{rcurve} in $x$. 

Suppose we are given a data set that has $K$ time steps consisting of zero-coupon spreads for all available tranches and maturities. In each time step of the sample period, the data set is represented by an $(I\times J)$ matrix where $I$ denotes the number of different time to maturities. Having the data set and the model, the estimation procedure comprises the specification of the model parameters in such a way that the model describes the whole data series as much as possible. Under the current modeling setup, the difficulty we face in estimation is due to the unobservability of the factor process. One way to overcome this problem is to use filtering. In a framework where the unobserved factor is a Gaussian process, a Kalman filter yields the exact likelihood function via providing the prediction error and its variance (see  \citet{H}). When using non-Gaussian models, however, the exact likelihood function is not available in most cases. In such a situation, one can use a quasi-maximum likelihood (QML) approach in which the idea is to substitute the exact transition density of the non-Gaussian factor by a Gaussian density with mean and variance being equal to the first two true moments of the factor process. This has been a widely used method especially for the estimation of affine term structure models (see, for example,  \citet{GP}, \citet{CS} and \citet{DS}). Both in the one and two-factor models we presented above, the factor process is non-Gaussian. Hence, to estimate the model parameters and obtain the unobservable factor series we use a QML approach based on a Kalman filter. Since the one-factor model is  nested within the two-factor model, in what follows we will only give the estimation procedure for the later one.

Here, we consider the partition $0=t_0<t_1<\cdots<t_K=T$ on the interval $[0,T]$ and denote the  value of the factor process at time $t_k$ by $(Y_{t_k},Z_{t_k})$. In Kalman filtering, there is the \emph{measurement (observation) equation} expressing the observed data as the sum of a linear function of the unobservable factor and a measurement error. The discrete time evolution of the unobservable factor at $t_k$ is, in turn, expressed by the \emph{transition equation} as linear in $(Y_{t_{k-1}},Z_{t_{k-1}})$. Inserting \eqref{p2}  into \eqref{dcurve} reveals that $R(t_k,\tau,i)$ is not linear in $(Y_{t_k},Z_{t_k})$. In order to obtain a linear measurement equation, we approximate the function $\beta_y$ as follows:
\begin{align*}
\beta_y(x)&\approx \sum_{j=1}^6 \beta_j 1_{[x_{j-1},x_j)}(x)
\end{align*}
where the coefficients are given by
\begin{align*} \beta_j=\frac{1}{x_j-x_{j-1}}\int_{x_{j-1}}^{x_j}\beta_y(x)dx.\end{align*}
This immediately yields
$B_z(\tau,x)=B_z(\tau,x_j)$ for $x\in [x_{j-1},x_j)$, implying that 
\[\int_{x_{j-1}}^{x_j} e^{-A(\tau,x)}dx=e^{-\kappa_z\theta_z\int_0^{\tau}B_z(s,x_j)ds}\int_{x_{j-1}}^{x_j}e^{-\alpha(x)\tau}dx \]
which finally provides the desired linear \emph{measurement equation}. For the $i^{th}$ time to maturity, $\tau_i$, $i=1,\cdots, I$, and tranche $j$, $j=1,\cdots,J$, the measurement equation reads:
\begin{equation}\nonumber
R(t_k,\tau_i,j)=C_z(\tau_i,j)+\frac{1}{\tau_i} (B_y(\tau_i,x_j)Y_{t_k}+B_z(\tau_i,x_j)Z_{t_k})+\epsilon(t_k,j)
\end{equation}
where
\[C_z(\tau_i,j)=\frac{1}{\tau_i}\log(x_j-x_{j-1})+\frac{1}{\tau_i}\kappa_z\theta_z\int_0^{\tau_i}\!\!B_z(s,x_j)ds-\frac{1}{\tau_i}\log\left(\!\!
                                                                \begin{array}{c}
                                                                  \int_{x_{j-1}}^{x_j}e^{-\gamma(e^{-a_0(x\wedge1)}-e^{-a_0})\tau_i}dx \\
                                                                \end{array}
                                                              \!\!\right)\]
and measurement errors, $\epsilon(t_k,j)$, are assumed to be i.i.d.\ $N(0,\sqrt{h_j})$, for some $h_j>0$, showing that the variance of the error depends on the tranche $j$ only. Now we define a new index $l(i,j)=(j-1)I+i$ and introduce $R_{t_k}$, the $(IxJ)$-dimensional vector which has $R(t_k,\tau_i,j)$ in its $l(i,j)$th entry. Furthermore, we denote by $H$ the corresponding $(IxJ)\times (IxJ)$ diagonal covariance matrix of observation errors which, for $i=1,\cdots,I$, has $h_j$ as the $l(i,j)$th diagonal entries.

Let $\mathbb{P}(Y_{t_k},Z_{t_k}|Y_{t_{k-1}},Z_{t_{k-1}})$ designate the transition density, which is the probability density of the factor at time $t_k$ given its value at time $t_{k-1}$. In line with the QML approach, we substitute the exact transition density of the factor by a Gaussian density, that is \[ \mathbb{P}(Y_{t_k},Z_{t_k}|Y_{t_{k-1}},Z_{t_{k-1}})\sim N(\mu_{t_k},Q_{t_k})\]  where the conditional mean $\mu_{t_k}$ and the covariance matrix $Q_{t_k}$ are distributed in such a way that the first moments of the approximate Normal and the exact transition density are equal. As the next step  we compute $\mu_{t_k}$ and $Q_{t_k}$ in the following proposition. This proposition mainly uses the fact that the factor process is affine and provides the desired expressions by utilizing the polynomial preserving property of affine processes.
\begin{prop}\label{condmoment} Suppose the process $(Y,Z)$ satisfies the dynamics given in \eqref{Y_P}-\eqref{Z_P}, then the $\mathbb{P}$-conditional expectation of $Y_t$ and $Z_t$  is in the following form:
\begin{align*}\begin{split}
 E[Y_t|Y_0\!=\!y,Z_0\!=\!z]&=\!\displaystyle\frac{\theta_z}{\kappa_z\!-\!\kappa_y}\!\Big(\!\kappa_z(1\!-\!e^{\!-\kappa_y t})\!-\!\kappa_y(1\!-\!e^{\!-\kappa_z t})\!\Big)\!+\!e^{\!-\kappa_y t}y
\!+\!e^{\!-\kappa_z t}\!\frac{\kappa_y}{\kappa_z\!-\!\kappa_y}\!\Big(\!e^{\!t(\kappa_z-\kappa_y)}\!-\!1\!\Big)z, \end{split}\\
E[Z_t|Y_0\!=\!y,Z_0\!=\!z]&=\!\theta_z(1\!-\!e^{\!-\kappa_zt})\!+\!e^{\!-\kappa_z t}z.
\end{align*}
Moreover,  the conditional variances $V_y$, $V_z$ and the conditional covariance $V_{yz}$ are given by:
\begin{align} \label{Vy}\begin{split}
V_y(t,y,z)&=\Bigg(e^{-(5\kappa_z+7\kappa_y)t}\Big(e^{(5\kappa_z+7\kappa_y)t}(\kappa_z-2\kappa_y)(\kappa_z-\kappa_y)^2(\kappa_z(\kappa_z\!+\!\kappa_y)\sigma_y^2+\kappa_y^2\sigma_z^2)\theta_z\\
&-2e^{(5\kappa_z\!+\!6\kappa_y)t}\kappa_z(\kappa_z\!-\!2\kappa_y)(\kappa_z^2-\kappa_y^2) \sigma_y^2(\kappa_z(\theta_z-y)+\kappa_y(y-z))+e^{(3\kappa_z+7\kappa_y)t}(\kappa_z\\ &-2\kappa_y)\kappa_y^3(\kappa_z+\kappa_y)\sigma_z^2(\theta_z-2z+2e^{(4\kappa_z+7\kappa_y)t}
\kappa_y^2(\kappa_y^2-\kappa_z^2) (\kappa_z(\sigma_y^2-2\sigma_z^2)+2\kappa_y\sigma_z^2)\\
&\times(\theta_z\!-\!z)-4e^{(4\kappa_z\!+\!6\kappa_y)t}\kappa_z(\kappa_z\!-\!2\kappa_y)\kappa_y^2\sigma_z^2(\kappa_z\theta_z\!-\!(\kappa_z+\kappa_y)z)+e^{5(\kappa_z\!+\!\kappa_y)t}\kappa_z(\kappa_z\!+\!\kappa_y)\\
&\times(\kappa_z^3\sigma_y^2(\theta_z-2y)-2\kappa_z^2\kappa_y\sigma_y^2(\theta_z-4y+z)+2\kappa_y^3(2\sigma_y^2y-\sigma_y^2z+\sigma_z^2z)\\
&+\kappa_z\kappa_y^2(-\sigma_z^2\theta_z+\sigma_y^2(\theta_z-10y+4z)))\Big)\Bigg)
\Bigg/\Bigg(2\kappa_z(\kappa_z\!-\!2\kappa_y)(\kappa_z\!-\!\kappa_y)^2\kappa_y(\kappa_z\!+\!\kappa_y)\Bigg),\end{split}\end{align}

\begin{align}
\label{Vz}V_z(t,y,z)&=\displaystyle\frac{\sigma_z^2 e^{-2\kappa_z t}(e^{\kappa_z t}-1)((e^{\kappa_zt}-1)\theta_z+2z)}{2\kappa_z},
\end{align}
 \begin{align}
\nonumber V_{yz}(t,y,z)&=\frac{e^{-(2\kappa_z+\kappa_y)t}\sigma_z^2}{2(\kappa_z^3-\kappa_z\kappa_y^2)} \Big(e^{(2\kappa_z+\kappa_y)t}(\kappa_z-\kappa_y)\kappa_y\theta_z-e^{\kappa_yt}\kappa_y(\kappa_z+\kappa_y)(\theta-2z)\\
\label{Vzy}&-2e^{(\kappa_z+\kappa_y)t}(\kappa_z^2-\kappa_y^2)(\theta_z-z)
 +2e^{\kappa_zt}\kappa_z(\kappa_z\theta_z-(\kappa_z+\kappa_y)z)\Big).
\end{align}
\end{prop}
\begin{proof} See Appendix.
\end{proof}

We are now ready to give the \emph{transition equation} implied by the two-factor model. Denote the time increment by $\delta t=t_k-t_{k-1}$ and define
 
\begin{align*}
M_0(t_k)&=\left(
         \begin{array}{c}
         \frac{\theta_z}{\kappa_z-\kappa_y}
         \left(\kappa_z(1-e^{-\kappa_y \delta t})-\kappa_y(1-e^{-\kappa_z \delta t})\right)\\
           \theta_z(1-e^{-\kappa_z \delta t})\\
         \end{array}
       \right),\\
    M_1(t_k&=\left(
           \begin{array}{cc}
            e^{-\kappa_y \delta t} & e^{-\kappa_z \delta t}\frac{\kappa_y}{\kappa_z-\kappa_y}\left(e^{\delta t(\kappa_z\!-\!\kappa_y)}-1\right) \\
             0 & e^{-\kappa_z \delta t} \\
           \end{array}
         \right).
      \end{align*}
The unobserved state at time $t_k$ is evolved from the previous state according to: 
\begin{align*}
\left(
  \begin{array}{c}
    Y_{t_k} \\
    Z_{t_k} \\
  \end{array}
\right)
&=M_0(t_k)+M_1(t_k)\left(
  \begin{array}{c}
    Y_{t_{k-1}} \\
    Z_{t_{k-1}} \\
  \end{array}
\right)+v_{t_k},\ \ 
\end{align*}
where $v_{t_k}$ are i.i.d. $N(0,Q(t_k))$ with the covariance matrix
\begin{align*}
 Q(t_k)&=\left(
         \begin{array}{cc}
           V_y\left(\delta t,Y_{t_{k-1}},Z_{t_{k-1}}\right) & V_{yz}\left(\delta t,Y_{t_{k-1}},Z_{t_{k-1}}\right) \\
           V_{yz}\left(\delta t,Y_{t_{k-1}},Z_{t_{k-1}}\right) & V_z\left(\delta t,Y_{t_{k-1}},Z_{t_{k-1}}\right) \\
         \end{array}
       \right)
\end{align*}
and $V_y$, $V_z$ and $V_{yz}$ are as given in \eqref{Vy}, \eqref{Vz} and \eqref{Vzy} respectively. 

\par Letting $t \rightarrow \infty$ in the conditional moments given in Proposition \ref{condmoment} yields the unconditional moments of the factor process which are given by the following corollary.
\begin{cor} \label{uncondm}Unconditional mean, variance and covariance of $Y_t$ and $Z_t$ is given by:
\begin{align*}
\mu_y^0&=\theta_z,\quad \mu_z^0=\theta_z,\\
V_y^0&=\displaystyle\frac{\sigma_y^2\theta_z}{2\kappa_y}+\displaystyle\frac{\kappa_y\theta_z\sigma_z^2}{2(\kappa_z+\kappa_y)\kappa_z},\\
V_z^0&=\frac{\sigma_z^2\theta_z}{2\kappa_z},\\
V_{yz}^0&=\displaystyle\frac{\sigma_z^2\theta_z\kappa_y}{2\kappa_z(\kappa_z+\kappa_y)}=\displaystyle\frac{\kappa_y}{(\kappa_z+\kappa_y)}V_z^0.
\end{align*}
\end{cor}
Denoting the transpose of a matrix $A$ by $A^{\top}$, we use the notation 
$\left(
    Y_{t_m|t_{n}},
    Z_{t_m|t_{n}} 
 \right)^{\top}$ to represent the estimate of the state of the factor process at $t_m$ given observations up to, and including time $t_n$. In the same vein the estimate for the error covariance matrix, which serves as a measure for the precision of the state estimate, will be denoted by $P_{t_m|t_n}$. In what follows we shed light on the filtering algorithm. 
 
Given the parameter set $\varphi=(\kappa_z,\kappa_y, \theta_z,\lambda^z,\lambda^y, \sigma_z,\sigma_y,a_0,\gamma, b_0,c_0, H)$, the Kalman filter consists of prediction and updating steps which are applied for each time point in the data sample in the following way:
 
 \begin{itemize}
 \item[1.]{\it Initialize} the filter by using the unconditional moments:
 \begin{align*}
 \left(\!
  \begin{array}{c}
    Y_{0|0} \\
    Z_{0|0} 
  \end{array}
\!\right)&=\!\left(\!
  \begin{array}{c}
    \theta_z \\
    \theta_z 
  \end{array}
\!\right),\quad  P_{0|0}=\!\left(\!\begin{array}{cc}
                       \displaystyle \frac{\sigma_y^2\theta_z}{2\kappa_y}\!+\!\frac{\kappa_y\theta_z\sigma_z^2}{2(\!\kappa_z\!+\!\kappa_y\!)\kappa_z}&\displaystyle \frac{\sigma_z^2\theta_z\kappa_y}{2\kappa_z(\!\kappa_z\!+\!\kappa_y\!)}\\
                        \displaystyle\frac{\sigma_z^2\theta_z\kappa_y}{2\kappa_z(\!\kappa_z+\kappa_y\!)} & \displaystyle\frac{\sigma_z^2\theta_z}{2\kappa_z}
                      \end{array}\right)
\end{align*}
\item[2.]{\it Prediction:} produce an estimate of the current state and covariance of the estimate via
\end{itemize}
 \begin{gather*}
 \left(\!
  \begin{array}{c}
    Y_{t_k|t_{k-1}} \\
    Z_{t_k|t_{k-1}} \\
  \end{array}\!
\right)=\!M_0(t_k)\!+\!M_1(t_k)\!\left(\!
  \begin{array}{c}
    Y_{t_{k-1}|t_{k-1}} \\
    Z_{t_{k-1}|t_{k-1}} \\
  \end{array}
\!\right)
\, \text{and}\, P_{t_k|t_{k-1}}\!=\! M_1(t_k)P_{t_{k-1}|t_{k-1}}M_1(t_k)^{\top}\!\!+\!Q(t_k),
 \end{gather*}
 respectively.
\begin{itemize}
 \item[3.]{\it Updating:} for all $i$, $i=1,\cdots,I$, and $j$, $j=1,\cdots,J$, compute
 \begin{align*}
 R(t_k|t_{k-1},\tau_i,j)&=\!C_z(\tau_i,j)\!+\!\displaystyle\frac{1}{\tau_i}\left(\!B_y(\tau_i,x_j)Y_{t_k|t_{k-1}}\!+\! B_z(\tau_i,x_j)Z_{t_k|t_{k-1}}\!\right),
\end{align*}
and form the corresponding $IxJ$-dimensional row vector $R_{t_k|t_{k-1}}$. Next, compute $R_{t_k}$ and the \emph{innovation vector} successively as
\begin{align*}e_{t_k}&=R_{t_k}-R_{t_k|t_{k-1}}.\end{align*}
Then, generate an $(IxJ)\times 2$ matrix, $B$, having $\displaystyle\frac{B_y(\tau_i,x_j)}{\tau_i}$ and $\displaystyle\frac{B_z(\tau_i,x_j)}{\tau_i}$ in the $\left(l(i,j),1\right)$st and $\left(l(i,j),2\right)$nd entries, respectively. Now, compute the \emph{innovation covariance} matrix, $F_{t_k}$ via
 \[F_{t_k}=BP_{t_k|t_{k-1}}B^{\top}+H. \]
After this, calculate the \emph{Kalman gain} with the following: 
\begin{align*}
K_{t_k}&=P_{t_k|t_{k-1}}B^{\top}F_{t_k}^{-1}.
\end{align*}
Lastly, update the state vector and the estimate covariance by 
\begin{gather*}
 \left(
  \begin{array}{c}
    Y_{t_k|t_{k}} \\
    Z_{t_k|t_{k}} \\
  \end{array}
\right)=\left(
  \begin{array}{c}
    Y_{t_k|t_{k-1}} \\
    Z_{t_k|t_{k-1}} \\
  \end{array}
\right)+K_{t_k}e_{t_k} \quad \text{and}\quad
 P_{t_k|t_k}=P_{t_k|t_{k-1}}-K_{t_k}BP_{t_k|t_{k-1}},
\end{gather*}
respectively.
\end{itemize}
The Kalman filter provides the following likelihood function:
\begin{equation}\nonumber
\log L(R_1,R_2,...,R_N;\varphi)=-\frac{K}{2}\log2\pi-\frac{1}{2}\sum_{k=1}^K\log|F_{t_k}|-\frac{1}{2}\sum_{k=1}^Ke_{t_k}^{\top}F_{t_k}^{-1}e_{t_k}.
\end{equation}
Notice that $L$ is a function of $e_t$ and $F_t$ which, eventually, depend on the parameter set $\varphi$. Thus, as the last step of the QML method, we choose $\varphi$ in such a way that the likelihood function is maximized. Finally, we want to point out that the observed data vectors may change size over the sample period. This is due to the unavailability of the data for some tranches and/or maturities in some days of the sample period. To overcome this problem, we adjust the Kalman filter algorithm in such a way that it takes the changes in the size of the data into account.
\section{Data Implementation}\label{daI}
In this section one and two-factor models described above are implemented on the real market data. Moreover, we perform an in-sample hedging analysis. We also run a simulation where normal and extreme loss scenarios are generated via method of importance sampling. Finally we assess the hedging performance
under these more general scenarios.

\subsection{Data Description}
The raw data comprises daily observations of iTraxx Europe from 30 August 2006 to 3 August 2010. The stripped data, which has been sourced from Bank Austria\footnote[1]{We thank Peter Schaller for providing the data.}, is the zero-coupon spreads, $R(t_k,\tau_i,j)$, for four different time to maturities $(\tau_1,\cdots,\tau_4)=(3,5,7,10)$ and six tranches $j=1,...,6$ with standard attachment and detachment points  $0\%$, $3\%$, $6\%$, $9\%$, $12\%$, $22\%$, $100\%$. This corresponds to $K=972$ observation days in each of which we have a $4\times 6$ observation matrix. 

We illustrate the time series of 5-year zero-coupon spreads and the index spreads across four different maturities in Figure~\ref{badata} and Figure~\ref{indexdata}, respectively. Naturally, market conditions are reflected in the data set. The index and tranche data follow relatively stable pattern from the beginning of the data period to July 2007, where the  credit crunch erupted. In March 2008, we observe a spike in the spread data which stems from the panic due to the possibility of the collapse of Bear Stearns. Furthermore, a drastic upward movement is observed starting from September 2008. This time period corresponds to the breakdown of the credit market due to events such as the bankruptcy of Lehman Brothers. Moreover, Figure~\ref{badata} and Figure~\ref{indexdata} together show that the tranche data and the index data have the same up and downward trends during the time period considered. One other feature of the data set we use is that there is no default event during the sample period. 

\begin{figure}[H]
 \centering
 \subfloat[iTraxx Europe 5-year zero-coupon spread data for all tranches.]{\label{badata}\includegraphics[height=5.5cm,width=7.5cm]{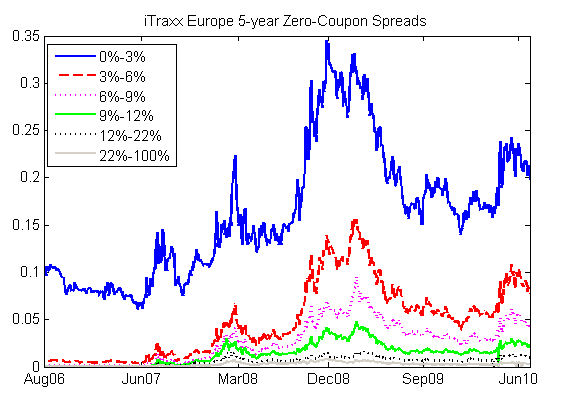}}
  \subfloat[iTraxx Europe index spread data for all maturities.]{\label{indexdata}\includegraphics[height=5.5cm,width=7.5cm]{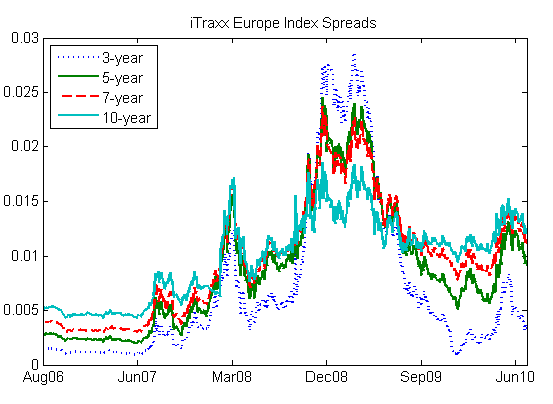}}
 \caption{ iTraxx Europe data from 30 August 2006 to 3 August 2010}
  \end{figure}

\subsection{Estimation Results}
Running the estimation algorithm given in Section~\ref{EM} we fit the one and two-factor models to the iTraxx data. During the analysis the risk-free rate is considered to be constant at $r=0.05$. In the following we  discuss the results of the empirical analysis.

As mentioned earlier, the QML approach makes it possible to estimate the model parameters and filter out the unobservable factor series simultaneously. We obtain the parameter estimates given in Table \ref{prm}. The likelihood values given in the table suggest that the corresponding likelihood ratio test statistic (LRT)\footnote{LRT is given as the $2(LL_2-LL_1)$ where $LL_2$ and $LL_1$ denote the value of the log-likelihood for the nested one-factor model and the two-factor model, respectively.} is highly significant and rejects the null hypothesis that the one-factor model is equivalent to the two-factor model. 

\begin{table}[htbp]
\centering
\begin{tabular}{|r|cc|cc|}
     \hline
        & \multicolumn{2}{c|}{1-factor} & \multicolumn{2}{c|}{2-factor}\\
   Parameter  & Estimate & SE & Estimate & SE \\
     \hline
    $\theta_z$&0.03& 5.7146 &0.0055&0.0002\\
        $\kappa_y$& -&-&0.5223 &0.0638\\
				$\kappa_z$&6.96e-05& 0.0130&0.2209& 0.0214\\
        $\sigma_y$& -& -&0.3806&0.0163\\
				$\sigma_z$& 0.15&0.0205 &0.2588&0.0202\\
        $\lambda_y$& -&-&-0.6671&0.0647\\
        $\lambda_z$& 1.44e-04& 0.0135&-0.3096&0.0235\\
        $a_0$&3.23e-05&1.2875e-005 &98.7846&3.6634\\
        $\gamma$&26.08&6.8044e+002 &0.3279&0.0130\\
        $b_0$&23.96& 0.2755 &26.4061&0.2371\\
        $c_0$& -&-&0.0911&0.0096\\
                       \hline
Log-likelihood&6.9365e+004&  &8.7842e+004&   \\     
\hline                    
                           \end{tabular}
													\caption{Parameter estimates and corresponding standard errors (SE) for the one-factor and two-factor affine models.}\label{prm} 
\end{table}
The filtered factor series are depicted in Figure ~\ref{Y_Z}. It is remarkable how the factor $Z$,  which drives the catastrophic component, stays almost zero until the breakout of the credit crisis. 
\begin{figure}[H]
\centering
 \includegraphics[height=5.5cm,width=8cm]{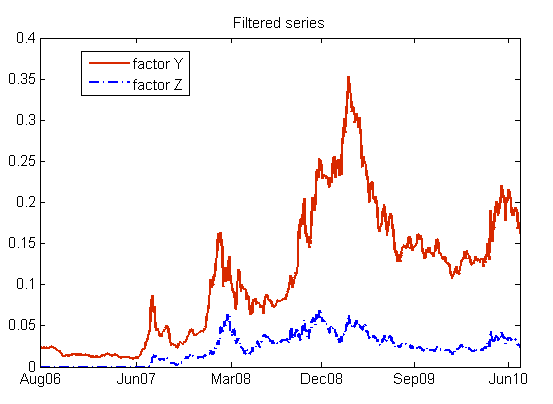}
\caption{  Filtered series of the factor $Y$ and $Z$ for the period 30 Aug 2006-3 Aug 2010.}
 \label{Y_Z}
\end{figure}

Taking $3-6\%$, $9-12\%$ and $22-100\%$ tranches as representatives for equity, mezzanine and senior tranches, respectively, we plot actual vs estimated zero-coupon spreads of 5-year maturity in Figure~\ref{ACvsEs1}. We observe that the two-factor model yields a plausible fit across all tranches, outperforming the one-factor model. Also, it is remarkable how the one-factor model estimates are far below the actual data whereas the two-factor model provides almost a perfect fit for the senior tranche. Here, we want to point out that a two-factor affine factor model with the restriction of a zero catastrophic component, as the one-factor-model is not able to fit the super-senior tranche. There, the importance of the catastrophic risk component of the two-factor model comes into play. That is, under the two-factor affine framework including the catastrophic component becomes inevitable for a model fit in senior tranches. 

\begin{figure}[H]
 \centering
 \subfloat{\label{f_2}\includegraphics[height=5cm,width=7.5cm]{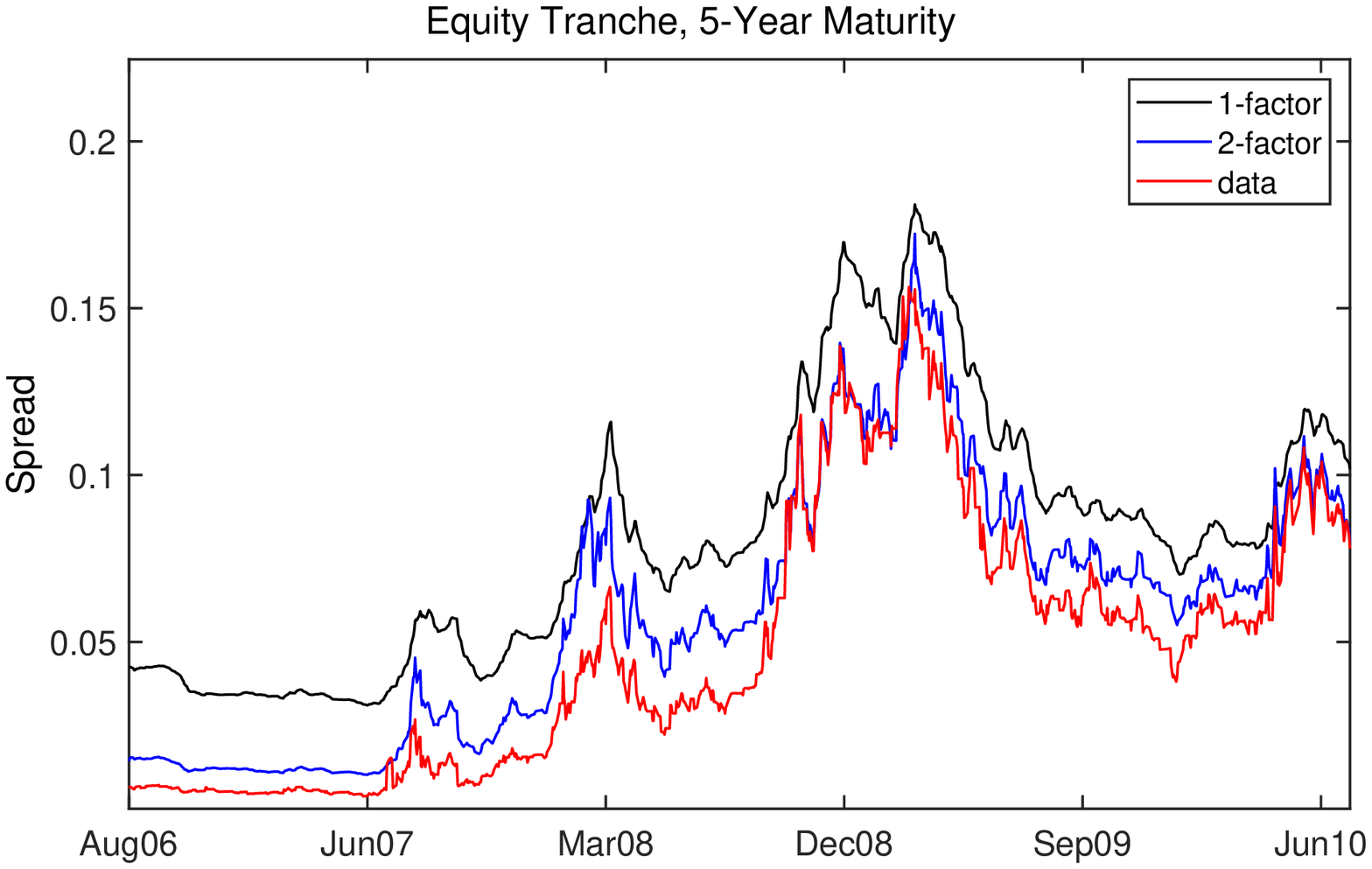}}
 \subfloat{\label{f_4}\includegraphics[height=5cm,width=7.5cm]{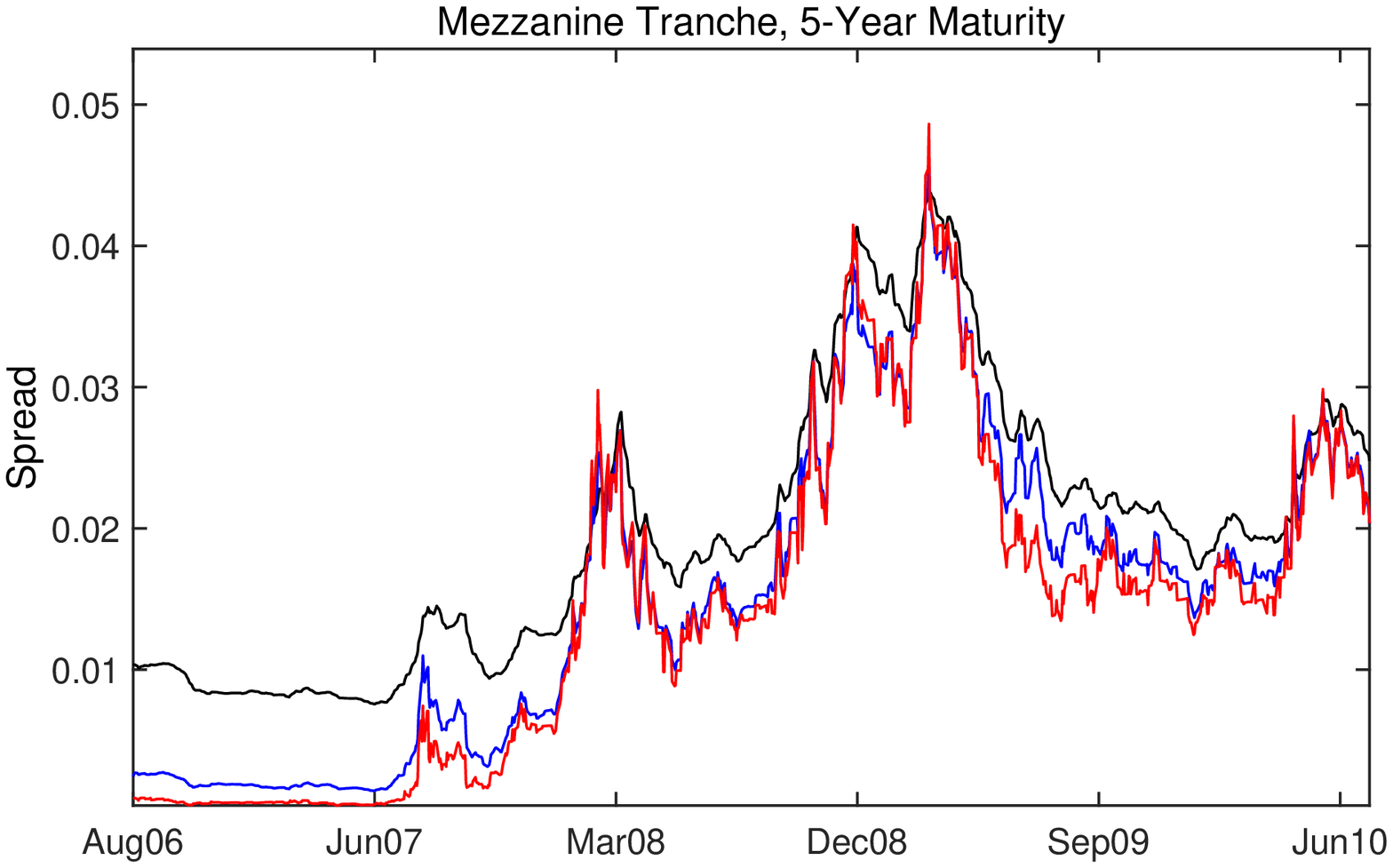}}\\
\subfloat{\label{f_6}\includegraphics[height=5cm,width=7.5cm]{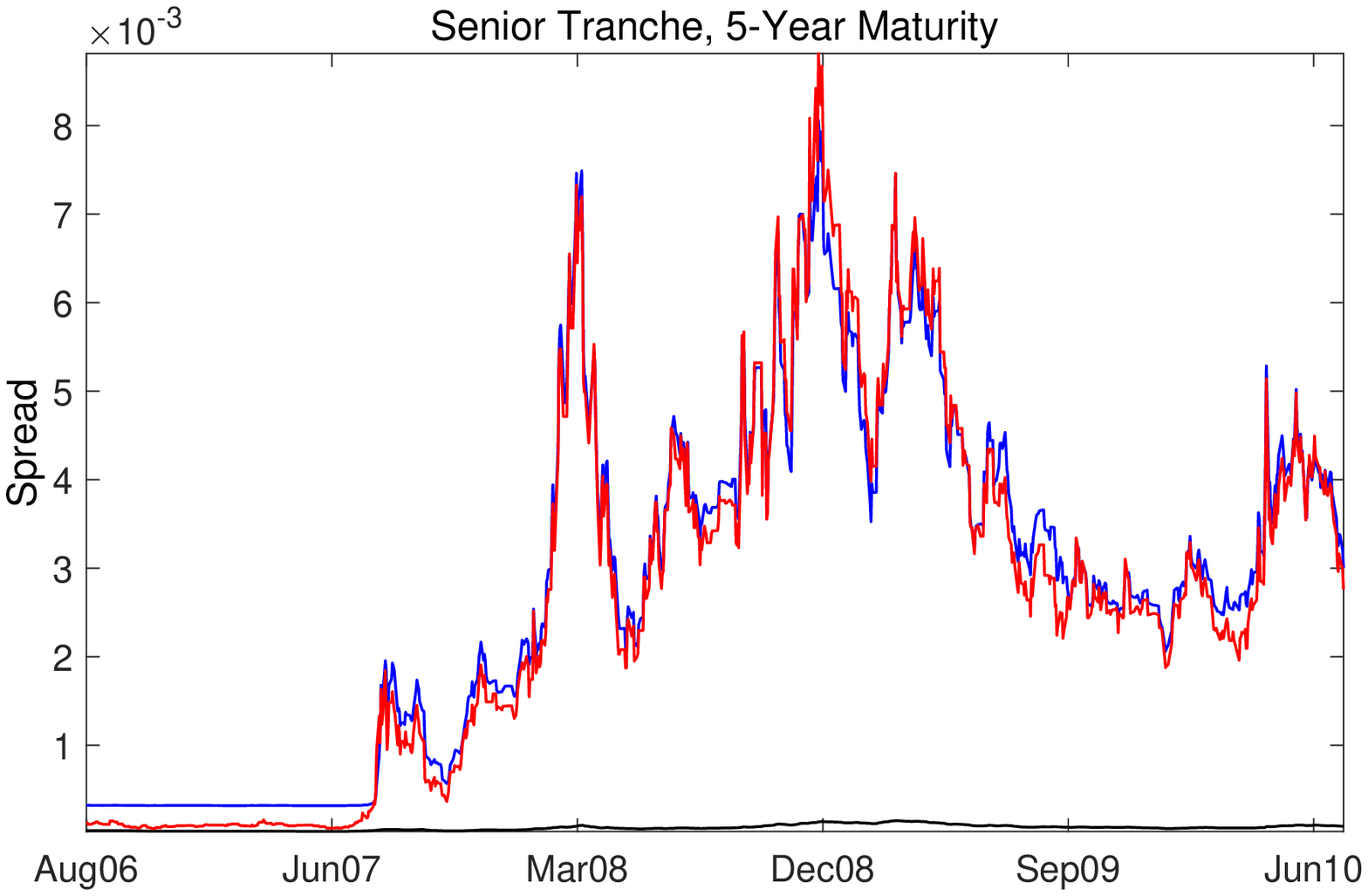}}
\caption{ Actual vs Estimated Spread Series}
\label{ACvsEs1}
\end{figure}

Given the parameter estimates and filtered factor series we simulate a loss trajectory. Inserting the parameter estimates, filtered factor series and the value of the loss trajectory into formula \eqref{ELG}, we obtain the implied expected loss given default series. Moreover, to investigate the effect of the catastrophic component, we fix the catastrophic risk parameter at $c_0=0$ and compute the corresponding expected loss given default series. Figure~\ref{ELGD} demonstrates how the implied expected loss given default with and without catastrophic component change in the sample period. Following from the fact that  the factor $Z$ stays very close to zero  by August 2007, the series with and without catastrophic component coincide in the corresponding part of the sample period. This stipulates that the catastrophic risk component is needed during the crisis times.  Hence, we conclude that considering a non-zero catastrophic component provides more flexibility in the modeling of expected loss given default in distressed markets.

\begin{figure}[H]
\centering
\includegraphics[height=5cm,width=8cm]{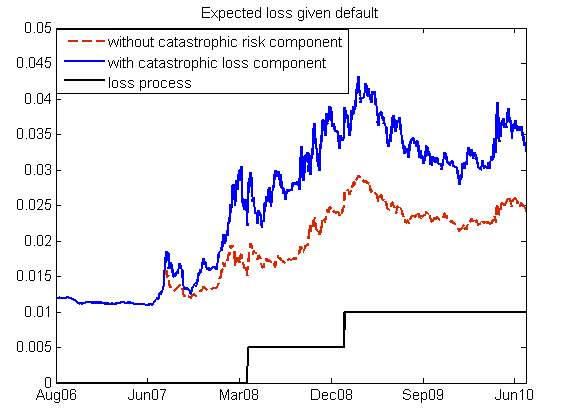}
\caption{Expected loss given default with and without catastrophic component}
  \label{ELGD}
\end{figure}

\subsection{In-Sample Hedging Results}
We first specify a hedging time horizon $[0,T]$, and trading days $0=t_0<t_1<\cdots t_K=T$ with $\delta t:=t_k-t_{k-1}\equiv 1/250$.

We denote by $CP=\{T_1,\dots,T_n\}$ the set of coupon payment dates and assume that any coupon payment date satisfies $T_i \in \{t_1,\dots,t_K\}$ for all $i$.

Assume a time series of the loss process $L_{t_k}$, and the spot value processes $\Gamma^{(x_1,x_2]}_{t_k}$ and $\Gamma^{(0,1]}_{t_k}$, $k=0,\dots,K$, are given. Consider a hedging strategy denoted by $\phi_{t_k}$. The resulting nominal value process $V_{t_k}$ of the self-financing hedging portfolio with zero initial capital, $V_{t_0}=0$, is given by the recursive formula
\begin{align*} V_{t_k}&= V_{t_{k-1}}e^{r\delta t}+PL_{t_k}^{(0,1]} \end{align*}
where  
\begin{equation}\nonumber
PL_{t_k}^{(0,1]}\!=\!\phi_{t_{k-1}}\left(\!\Gamma_{t_k}^{(0,1]}\!-\!\Gamma_{t_{k-1}}^{(0,1]} e^{r\delta t_k}\!+\!1_{\{t_{k}\in CP\}} \kappa_{t_0}^{(0,1]} H^{(0,1]}(L_{t_{k}})\!-\!\left(\!H^{(0,1]}(L_{t_{k-1}})-H^{(0,1]}(L_{t_{k}})\!\right) \!\right)
\end{equation}
indicates the nominal daily profit and loss on $(t_{k-1},t_k]$. The discounted gains process $G^{(0,1]}_{t_k}$ of holding the index is then given by $G^{(0,1]}_{t_k}=e^{-r t_k} V^{(0,1]}_{t_k}$ where $V^{(0,1]}_{t_k}$ denotes the nominal value process for the strategy $\phi_{t_k}\equiv 1$. And following the same way one can obtain $G^{(x_1,x_2]}_{t_k}$.

The hedging algorithm consists of computation and comparison of the nominal value process of the hedging portfolio and the gains process of the tranche position for each day of the hedging period. For the considered sample period,  we implement the hedging algorithm for all attachment and detachment points. Taking $3-6\%$, $9-12\%$ and $22-100\%$ tranches as representatives for equity, mezzanine and senior tranches, respectively, Figure~\ref{TR1} depicts the hedging strategies, and the series for the gains process and nominal spot value of the STCDOs as well as the hedging portfolio values. The change in $\phi$ at the beginning of the crisis around July 2007 is observed to exhibit different patterns for each tranche: while for equity tranches $\phi$ is decreasing in absolute value, indicating a reduction in insurance demand, for the senior tranche there is the opposite behavior. We attribute this change to the changing correlation structure between tranches and the index: running correlation series indicate an upward jump towards $1$ for the senior tranches and there is the opposite behavior concerning the equity tranche. The running correlation series stays relatively stable for the mezzanine tranches. When we focus on the corresponding hedging portfolio value, hedging strategy performs relatively better for the mezzanine tranches. Moreover, we observe a poor hedging performance for the equity tranches, in particular during the crisis period.  
\begin{figure}[H]
 \centering
 \subfloat[Equity tranche]{\label{tr2_1}\includegraphics[height=4.5cm,width=6cm]{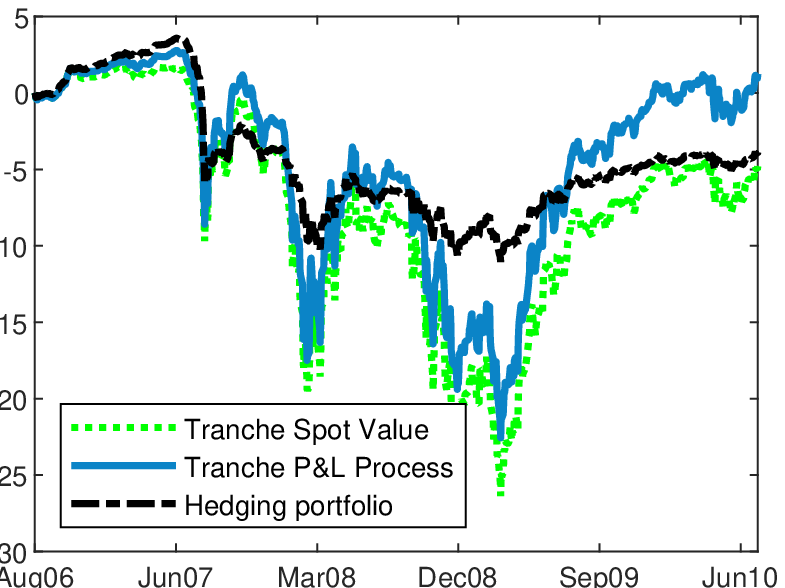}
\includegraphics[height=4.5cm,width=6cm]{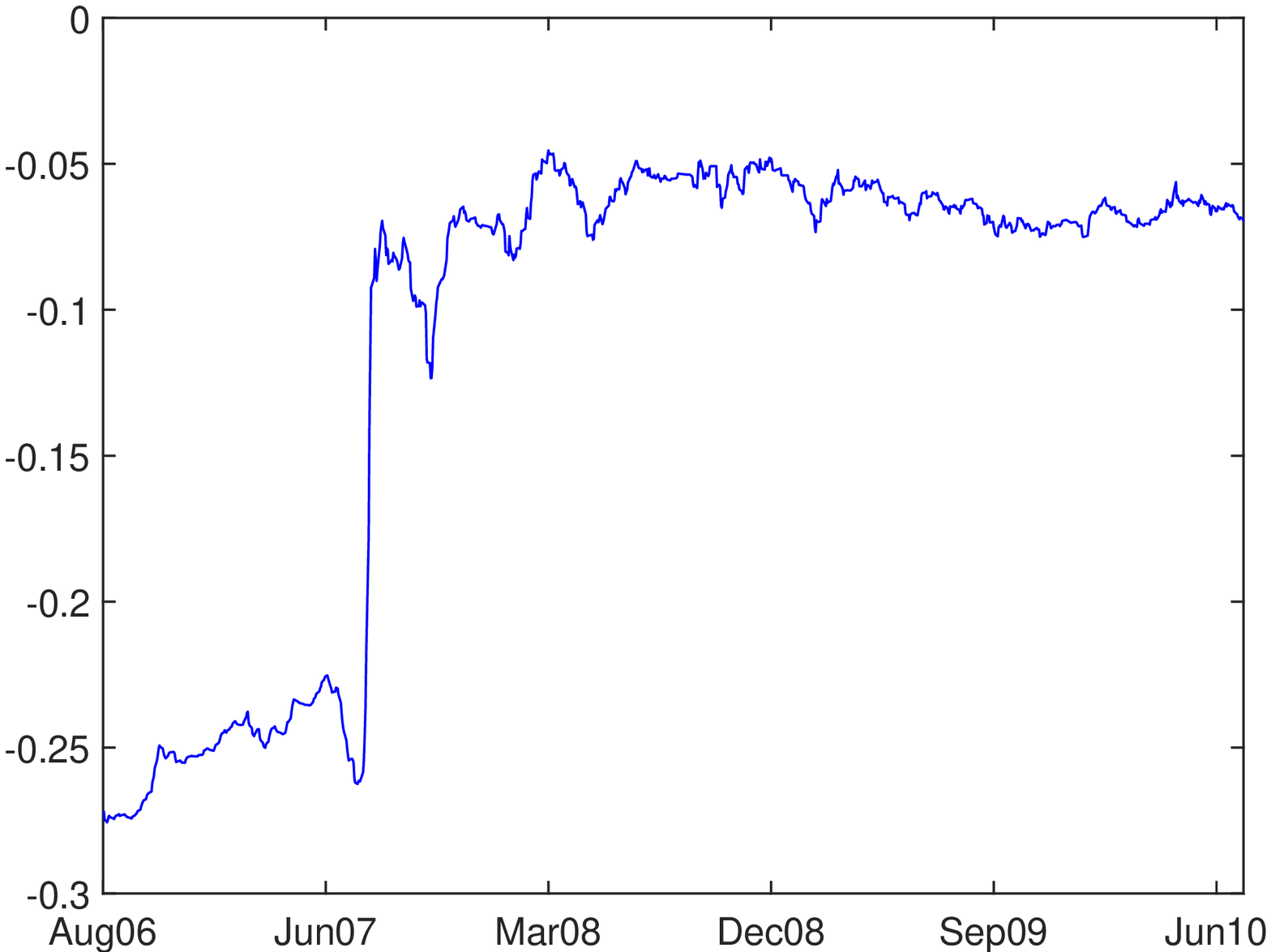}} \\
  \subfloat[Mezzanine tranche]{\label{tr4_1}\includegraphics[height=4.5cm,width=6cm]{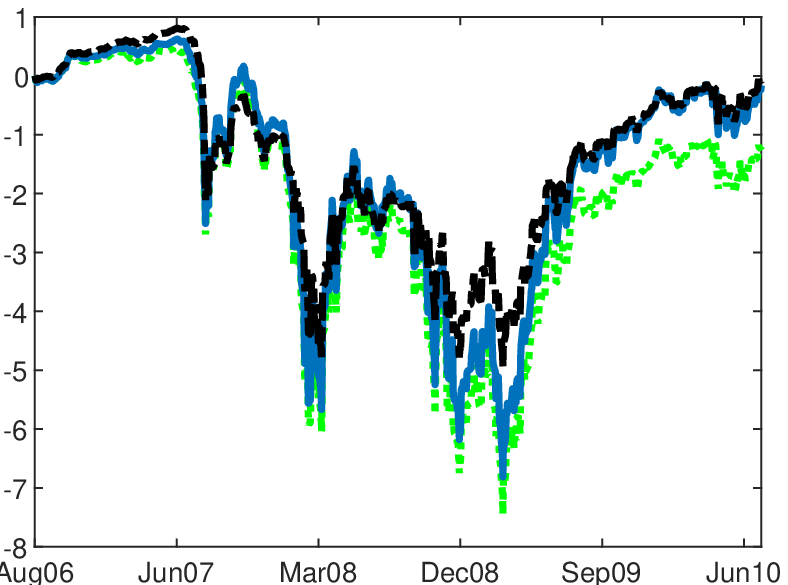} \includegraphics[height=4.5cm,width=6cm]{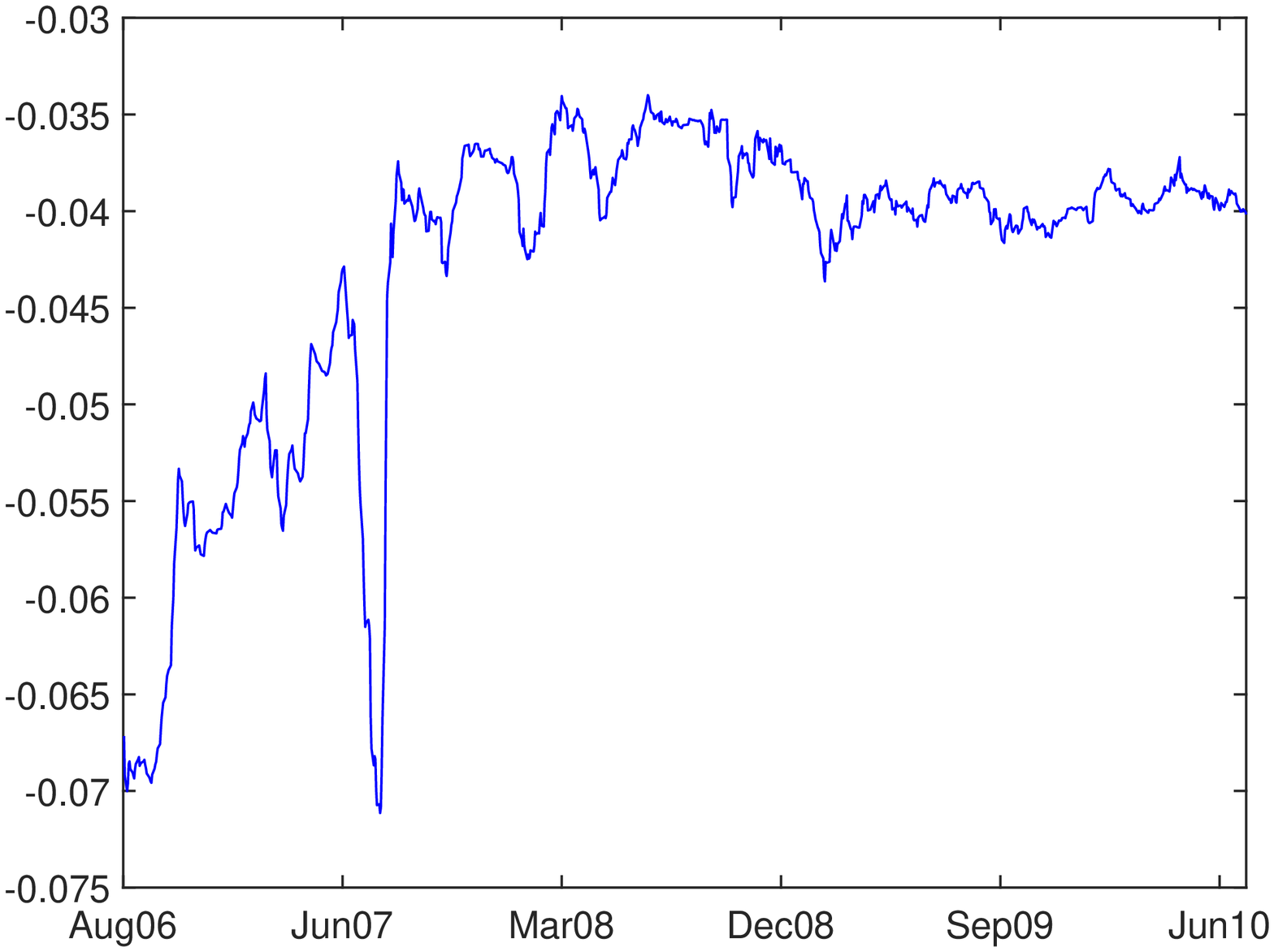}}\\
 \subfloat[Senior tranche]{\label{tr6_1}\includegraphics[height=4.5cm,width=6cm]{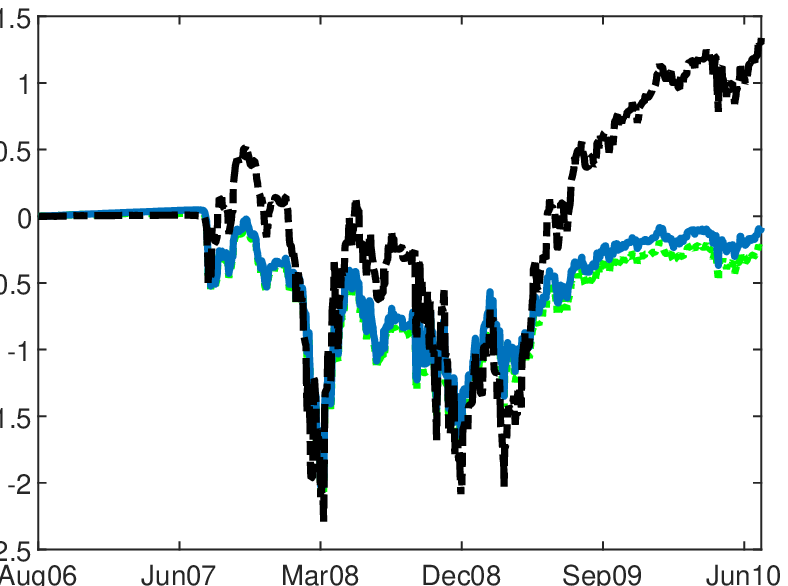} \includegraphics[height=4.5cm,width=6cm]{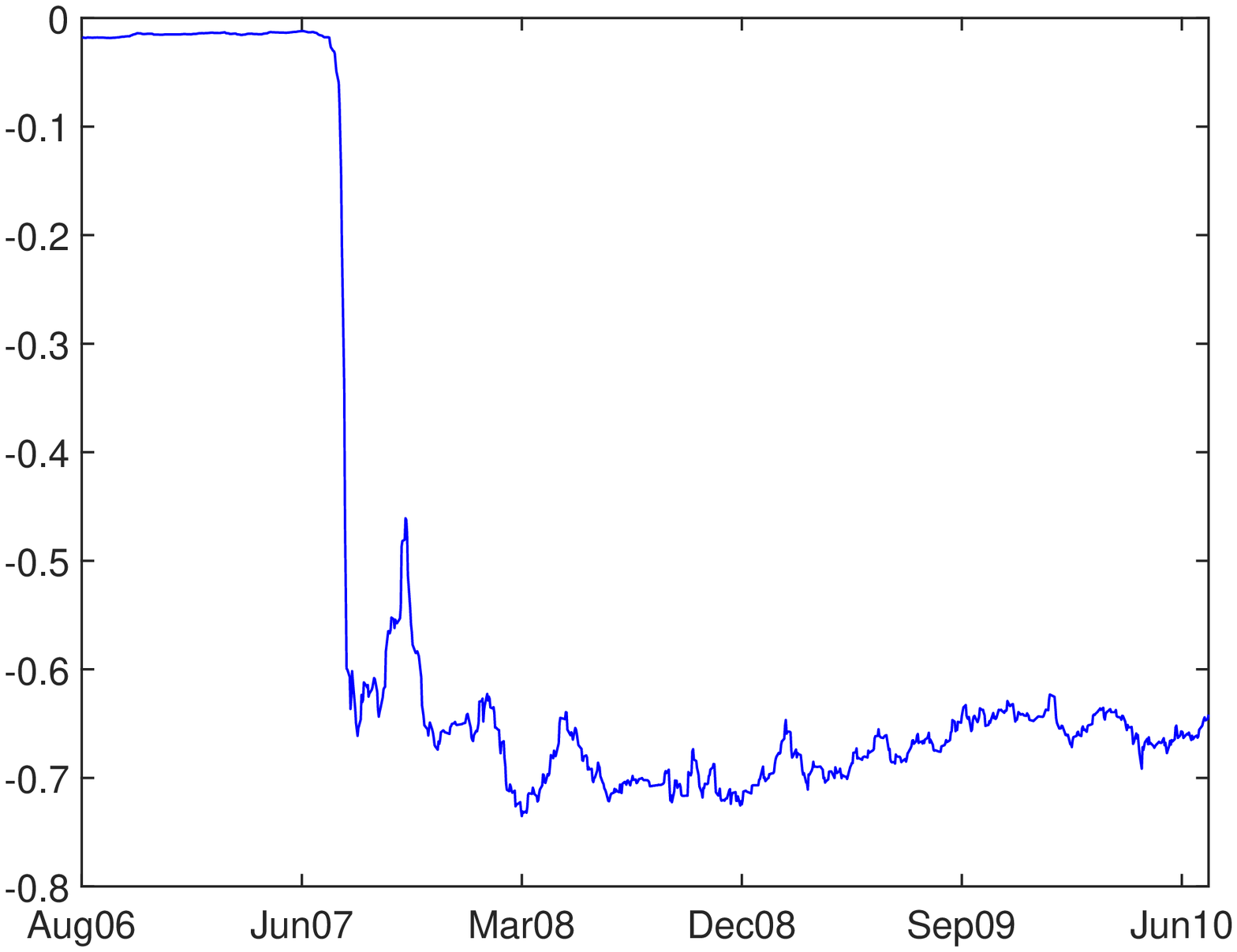}} 
 \caption{Hedging results for the sample period: gains process, nominal spot value and hedging portfolio value for equity, mezzanine and senior tranches (left); the hedging strategy $\phi$ (right).}
\label{TR1}
 \end{figure}

In order to better assess the performance of a given hedging strategy, \citet{CK} suggest  \emph{reduction in volatility} criteria which measures the reduction in the dispersion of the profit distribution with respect to the unhedged position. Formally, it is defined as the ratio of the volatility of daily profits from the hedged position to the volatility of the daily profits from the unhedged tranche position. According to this criteria, a hedging strategy performs better as long as the related reduction in volatility value is smaller. We report the reduction in the volatility of P\&L in Figure~\ref{RV}. According to this criteria, the hedging strategy performs better for the mezzanine tranches.
\begin{figure}[H]
 \centering
 \includegraphics[height=6cm,width=8cm]{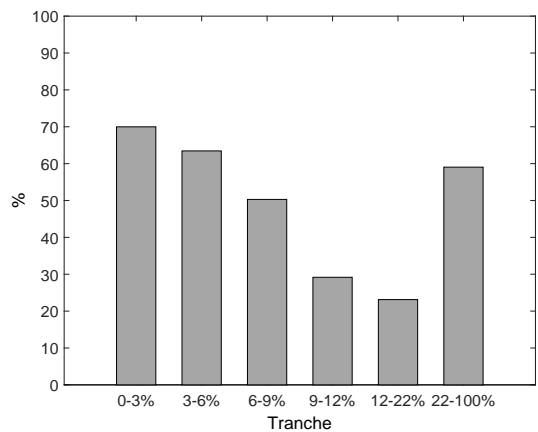}
\caption{Reduction in volatility }
\label{RV}
\end{figure}

\subsection{Simulation Study}\label{simul}
In the simulation analysis our objective is to elaborate on the performance of the variance minimizing hedging strategy in a more general framework, where scenarios with nonzero losses are permitted. Recall that in the current modeling setup we deal with three stochastic processes, namely the factors $Y$, $Z$ and the loss process $L$. We use Euler discretization to approximate the discrete time evolution of the factors Y and Z in Equation \eqref{Y_P}-\eqref{Z_P} on the equidistant time grid $0=t_0<t_1<...<t_K=T$ with $t_{k+1}-t_k\equiv \delta t=1/250$.

As the next step, to simulate the loss process $L$ we use the simulated factors $Y$, $Z$, the set of parameter estimates and an additional parameter $\Psi$ which we interpret as the \emph{importance sampling} parameter. The reason why we need the parameter $\Psi$ is as follows: there does not occur any default during the sample period we consider. Hence, the parameter set coming from the in-sample analysis is not able to generate a remarkable number of jumps. Moreover, the Monte Carlo simulation is known to fail in generating rare events unless the number of simulated scenarios is very large. Nevertheless, a frequently used technique in stress scenario generation is importance sampling (see, e.g. \citet[Section 2.7]{boyle} and references therein). In this context, it is possible to manipulate the number of jumps via amplifying the jump intensity given in \eqref{arrintense} through the parameter $\Psi$. However, one should take care of the necessary measure change for the adjustment of probabilities assigned to each scenario. This is necessary in particular for computing the empirical distribution function of the losses.  

We define the cumulative intensity process $\bar{\Lambda}_t=\int_0^t\Lambda_s ds$ and sketch the algorithm for simulating a loss trajectory of length $K$ as follows.
\begin{Algorithm}[Loss process simulation]

\begin{enumerate}
\item Initiate the jump time $\tau=0$, the number of jumps $N=0$, the loss process $L_{t_0}=0$, and the cumulative arrival intensity $\bar{\Lambda}_{t_0}=0$.
\item Generate a number $U$ from exponential distribution with parameter $1$.
\item While $k<K$ and $\bar{\Lambda}_{t_k}-\bar{\Lambda}_{\tau}<U$ calculate $\bar{\Lambda}_{t_{k+1}}$ via
\begin{align}
\nonumber \bar{\Lambda}_{t_{k+1}}&=\bar{\Lambda}_{t_k}+\Psi(\alpha(L_{\tau})+\beta_y(L_{\tau})Y_{t_k}+\beta_z(L_{\tau})Z_{t_k}-r)\delta t
\end{align}
\begin{align}
\nonumber \text{set} \   \  L_{t_{k+1}}=L_{t_k} , \ \  k \mapsto k+1
\end{align}
\item If $\bar{\Lambda}_{t_k}-\bar{\Lambda}_{\tau} \geq U$,  i.e., when a jump occurs generate a number $s$ from the standard uniform distribution. Compute the jump size via
\[ \Delta L_{t_k}=F_L^{-1}(L_{t_k},Y_{t_k},Z_{t_k},s)\]
where $F_L$ is the cumulative loss given default distribution given by
\begin{align}
\begin{split}
\nonumber F_L(L_{t_k},Y_{t_k},Z_{t_k},x)&=\frac{\alpha(L_{t_k})+\beta_y(L_{t_k})Y_{t_k}+\beta_z(L_{t_k})Z_{t_k}-\alpha(L_{t_k}+x)}{\alpha(L_{t_k})+\beta_y(L_{t_k})Y_{t_k}+\beta_z(L_{t_k})Z_{t_k}-r}\\
\nonumber \\
&-\frac{\beta_y(L_{t_k}+x)Y_{t_k}+\beta_z(L_{t_k}+x)Z_{t_k}}{\alpha(L_{t_k})+\beta_y(L_{t_k})Y_{t_k}+\beta_z(L_{t_k})Z_{t_k}-r}
\end{split}
\end{align}
Update the loss path, jump time and number of jumps
\[L_{t_k}=L_{t_k}+\Delta L_{t_k}, \ \ \tau=t_k, \ \  N=N+1\]
\item If $k<K$ return to 2, else stop.
\end{enumerate}
\end{Algorithm}
Employing the methodology described above, we simulate $2000$ scenarios, $1000$ of which are the normal scenarios and generated via taking importance sampling parameter $\Psi=1$. On the other hand, in order to simulate $1000$ stress scenarios we set  $\Psi=100$. We also set the probability estimate of each of the $1000$ normal scenarios equal to $q(i)=1/1000$, $i=1,2,...,1000$, while for the stress scenarios we adjust the probability estimate of each scenario in the following way. 

Denote by $\tau_n$ the $n^{th}$ jump time of the process $L$ . Changing the jump intensity of the process $L$ from $\Lambda_t$ to $\Psi \Lambda_t$ and leaving the jump size distribution unchanged is tantamount to an equivalent change of measure where the measure $\mathbb{P}^{\Psi}\sim \mathbb{P}$ is characterized by
\[\frac{d\mathbb{P}^{\Psi}}{d\mathbb{P}}|\mathcal{F}_t=M_t\] with the Radon-Nikodym derivative $M_t$ given by (see, e.g. \citet[Chapter VIII, T10]{bremaud})
\begin{align}\label{RNden}
M_t&=\left(\prod_{n\geq 1}\Psi 1_{\{\tau_n \le t\}}\right)e^{\int_0^t(1-\Psi)\Lambda_s ds}
= \left(\prod_{n\geq 1}\Psi 1_{\{\tau_n \le t\}}\right)e^{\bar{\Lambda}_t(1-\Psi)}.
\end{align}
Now suppose we generate the $i^{th}$ stress scenario from the distribution $\mathbb{P}^{\Psi}$. For this particular scenario we denote the total number of jumps realized in $[0,t_K]$ by $N^i$. Then, according to \eqref{RNden} we define the corresponding weight of the scenario $i$ under $\mathbb{P}$ by
\[w(i)=\frac{e^{(\Psi-1)\sum_{k=0}^{K-1}(\alpha(L_{t_k})+\beta_y(L_{t_k})Y_{t_k}+\beta_z(L_{t_k})Z_{t_k}-r)\delta t}}{\Psi^{N_i}}.\]
Due to the law of large numbers we have $\frac{\sum_{i=1}^{1000}w(i)}{1000}\approx 1$. Nevertheless, we normalize the weight $w(i)$ in an exact way and obtain the following estimate for the $\mathbb{P}$-probability of scenario $i$
\[p(i)=\frac{w(i)}{\sum_{i=1}^{1000}w(i)}.\]
Finally, to aggregate the $2000$ scenarios we give equal weight to normal and stress scenarios, and set probabilities $\bar{q}(i)=q(i)/2$ and $\bar{p}(i)=p(i)/2$, so that $\sum_{i=1}^{1000}\bar{q}(i)+\bar{p}(i)=1$, as it should be.
 Given the parameter estimates and the filtered factor series $(Y_{t_k}, Z_{t_k})$ we utilize formula \eqref{gamma} to construct series for the spot value of $5$-year index and $5$-year STCDO tranches. Then, we study the hedging of STCDOs with the index. Moreover, we perform a conditional simulation analysis in which $2000$ loss scenarios are generated conditional on the original  filtered factor trajectories. 
The results are given below.

 We take the set of simulated scenarios and focus on the final date, $T=250$, of the simulation period. The empirical cumulative loss distribution function (under $\mathbb{P}$)  at $T$ is given in Figure~\ref{CDF_L}. According to the figure, the simulation procedure is successful in the sense that it is able to produce loss scenarios ranging between $0\%$ and $10\%$.
\begin{figure}[H]
\centering
  \includegraphics[height=6cm,width=7cm]{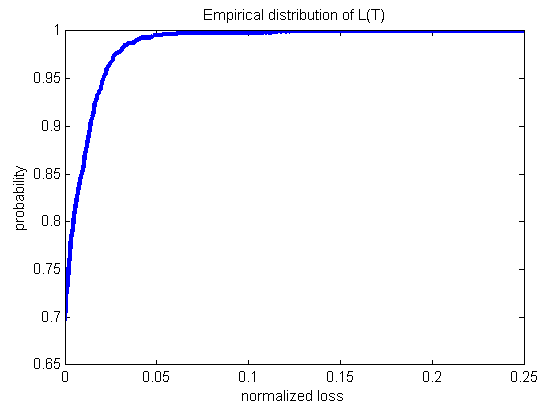}
  \caption{Empirical cumulative loss distribution at $T=250$ for a total number of 2000 scenarios.}
	\label{CDF_L}
\end{figure}
In Figure~\ref{PL_CDF} we depict the date $T$ empirical cumulative distribution function of the total hedging portfolio P\&L. This figure implies that for mezzanine and senior tranches, variance-minimizing hedging strategies yield normalized total portfolio P\&L values which are close to zero in most of the simulated trajectories. In other words, on average the hedging strategy performs well for the mezzanine and senior tranches.
\begin{figure}[H]
 \centering
 \includegraphics[height=6cm,width=8.5cm]{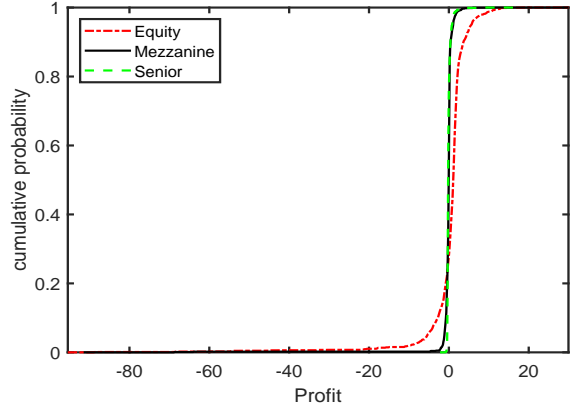}
\caption{ Empirical distribution of normalized total portfolio P\&L at $T$. }
\label{PL_CDF}
\end{figure}
Next, for each scenario realization we compute the reduction in volatility for all tranches and compute the descriptive statistics for reduction in volatility by taking the Radon-Nikodym densities of the stress scenarios into account. Table~\ref{RSD_varmin} illustrates the results: variance minimizing hedge is observed to yield the greatest reduction in variance for the mezzanine tranches. 

\begin{table}[htbp]
\setlength{\tabcolsep}{3pt}
\centering
\resizebox{9cm}{!} {
\begin{tabular}{|c||cccccc|}
\hline
&Mean&Median&Std&CV&Max&Min\\
\hline
Equity&64.95 &  55.72&43.56&  0.67&601.6&3.36  \\
Mezzanine&43.28&  36.01 &34.33 &0.79 &2223.6&10.26 \\
Senior&50.35 &  34.40 &62.19&1.23&3261.6&15.98 \\
\hline
\end{tabular}}
\caption{Descriptive Statistics of reduction in volatility for the variance-minimizing hedge}
\label{RSD_varmin}
\end{table}

We now present the results of the conditional simulation analysis. We fix the filtered factor series given in Figure~\ref{Y_Z} and conditional on these trajectories we simulate  2000 loss scenarios again with the importance sampling parameter values $\Psi=1$ and $\Psi=100$. Conditional distribution of the simulated loss process (under$\mathbb{P}$) at time $T$ is depicted in Figure~\ref{CCDF_L}.
One striking result is that, when compared with the loss distribution function given in Figure~\ref{CDF_L}, conditional loss distribution in Figure~\ref{CCDF_L} gives higher probability to losses greater than $10\%$. Moreover, simulation results suggest that for normal scenarios, $\Psi=1$, in 815 of 1000  simulated loss trajectories, there occurred a jump, that is, a default. These findings suggest that an actual default event in the iTraxx was very much likely to occur. 

\begin{figure}[H]
\centering
\includegraphics[height=5.5cm,width=7cm]{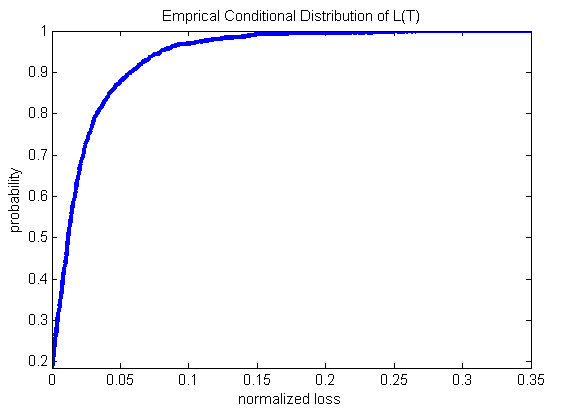}
\caption{Empirical conditional cumulative loss distribution function at $T=250$ for a total number of 2000 scenarios.}
\label{CCDF_L}
\end{figure}

\section{Conclusion}
In this study, we propose an affine two-factor model for the pricing and hedging of STCDOs. The most distinguishing feature of this model lies in the fact that a catastrophic risk component is considered as a tool for explaining the dynamics of the super-senior tranches. To test the real world performance we estimate the affine factor model on the iTraxx Europe data covering a period which witnessed different market conditions such as the recent credit crisis. As the main tool for the estimation of the affine factor model we use quasi-maximum likelihood based on a Kalman filter. This method requires the knowledge of conditional moments of the factor process. In this context, we utilize the polynomial preserving property for affine diffusion processes and compute the first two conditional moments of the factor process explicitly. Estimation results show that the two-factor model with the catastrophic component is successful in terms of fitting the market data even for super-senior tranches. Next, we compute the variance-minimizing hedging strategy based on the affine model. We investigate performance of the variance-minimizing hedging strategy on the data. We also ran a simulation analysis, in which the objective is to test the performance of the hedging strategy under more general loss scenarios.     

Our findings suggest that within the data period, According to the reduction in volatility criteria, the variance minimizing strategy is effective in reducing the risk of mezzanine tranches. The simulation study yields results of the same direction. 

\newpage
\begin{appendix}
\section{Appendix}\label{A1}
\begin{proof}[Proof of Proposition~\ref{condmoment}] Suppose we are given the process $X:=(Y,Z)$ where $Y$ and $Z$ solve \eqref{Y_P} and \eqref{Z_P}, respectively. This suggests that $X$ is an affine diffusion process with the state space $\mathcal{X}= \mathbb{R}_+^2$ (for a detailed information on affine diffusions see, e.g., \citet[Chapter 10]{filipo}). Now let $\tau \geq 0$, $k\in\mathbb{N}$, $(y,z)\in\mathcal{X}$ and denote the $k^{th}$ conditional  cross moment of $X_T$, $T=t+\tau$, by
\[f_k(\tau,Y_t,Z_t)=E[Y_{t+\tau}^pZ_{t+\tau}^q|Y_t=y,Z_t=z] \quad p,q\in\mathbb{N}, \ p+q=k. \]
An affine diffusion has the property that the $k^{th}$ conditional moment always exists and is a polynomial of at most degree $k$ of the current state $(y,z)$ (see Theorem 2.5 in \citet{filipovic2020polynomial} ). Being an affine diffusion, $X$ also possesses the Markov property (see, for instance \citet[Chapter III]{RY} for a detailed information on Markov processes). In particular, formally $f_k(\tau,y,z)$ solves the Kolmogorov backward equation
\begin{align}
\nonumber \frac{\partial}{\partial \tau}f_k(\tau,y,z)&=\mathcal{L}f_k(\tau,y,z),\\
\label{infgen1} f_k(0,y,z)&=y^pz^q
\end{align}
where $\mathcal{L}$ denotes the infinitesimal generator of the process  $X$ given by:
\[
\mathcal{L}=\kappa_y(z-y)\frac{\partial}{\partial y}+\kappa_z(\theta_z-z)\frac{\partial}{\partial z} +\frac{1}{2}\sigma_y^2y\frac{\partial^2}{ \partial y^2} +\frac{1}{2}\sigma_z^2z\frac{\partial^2}{\partial z^2}
.\]
To compute the conditional moments, taking the polynomial preserving property of process $X$ into account, one can use a polynomial ansatz in equation \eqref{infgen1}. Then, matching the coefficients yields a system of ordinary differential equations whose solution gives the coefficients of the polynomial in the ansatz. In the following, we follow this procedure as the first step towards the computation of the moments up to and including order two:

{\emph(i)} Let $E[Z_{t+\tau}|Y_t=y,Z_t=z]=:g(\tau,Y_t,Z_t)$. Function $g$ formally solves the Kolmogorov backward equation, that is,
\begin{equation}\label{g1}
\partial_{\tau}g=\kappa_y(z-y)\partial_y g+\kappa_z(\theta_z-z)\partial_z g+\frac{1}{2}\sigma_y^2y\partial_{yy}g+\frac{1}{2}\sigma_z^2z\partial_{zz}g
\end{equation}
Since  $X$ is  an affine process, we have the polynomial property of moments, that is, $g$ is in the following form
\begin{equation}\label{moment}
g(\tau,y,z)=g_0(\tau)+g_y(\tau)y+g_z(\tau)z
\end{equation}
for some functions $g_0$ ,$g_y$, $g_z$. Plugging \eqref{moment} into \eqref{g1} gives
\begin{align}
\nonumber \displaystyle\frac{d}{d\tau}g_{0}+\displaystyle\frac{d}{d\tau}g_{y}y+\displaystyle\frac{d}{d\tau}g_{z}z&=\kappa_y(z-y)g_{y}+\kappa_z(\theta_z-z)g_{z}
\end{align}
Comparing the coefficients on the right and left hand side, we get the following system of equations.
\begin{align}
\nonumber \displaystyle\frac{d}{d\tau}g_{0}&=\kappa_z\theta_zg_{z},\\
\nonumber g_{0}(0)&=0,\\
\nonumber \displaystyle\frac{d}{d\tau}g_{y}&=-\kappa_yg_{y},\\
\nonumber g_{y}(0)&=0,\\
\nonumber \displaystyle\frac{d}{d\tau}g_{z}&=\kappa_yg_{y}-\kappa_z g_{z},\\
\nonumber g_{z}(0)&=1.
\end{align}
Solving above system we get $g_{z}(\tau)=e^{-\kappa_z\tau}$, $g_{y}(\tau)\equiv0$ and $g_{0}(\tau)=\theta_z(1-e^{-\kappa_z\tau})$ implying that
\begin{equation}\label{Ez}
g(\tau,y,z)=\theta_z(1-e^{-\kappa_z\tau})+e^{-\kappa_z \tau}z
\end{equation}

{\emph(ii)} We set $E[Y_{t+\tau}|Y_t=y,Z_t=z]=:h(\tau,Y_t,Z_t)$. Formally, $h$ satisfies
\begin{equation}\label{g2}
\partial_{\tau}h=\kappa_y(z-y)\partial_y h+\kappa_z(\theta_z-z)\partial_z h+\frac{1}{2}\sigma_y^2y\partial_{yy}h+\frac{1}{2}\sigma_z^2z\partial_{zz} h
\end{equation}
From the polynomial property of moments again we have
\begin{equation}\label{moment2}
h(\tau,y,z)=h_{0}(\tau)+h_{y}(\tau)y+h_{z}(\tau)z.
\end{equation}
Plugging \eqref{moment2} into \eqref{g2} gives
\begin{align}
\nonumber\displaystyle\frac{d}{d\tau}h_{0}+\displaystyle\frac{d}{d\tau}h_{y}y+\displaystyle\frac{d}{d\tau}h_{z}z&=\kappa_y(z-y)h_{y}+\kappa_z(\theta_z-z)h_{z}.
\end{align}
Matching the coefficients we obtain
\begin{align}
\nonumber \displaystyle\frac{d}{d\tau}h_{0}&=\kappa_z\theta_zh_{z},\\
\nonumber h_{0}(0)&=0,\\
\nonumber\displaystyle\frac{d}{d\tau}h_{y}&=-\kappa_yh_{y},\\
\nonumber h_{y}(0)&=1,\\
\nonumber\displaystyle\frac{d}{d\tau}h_{z}&=\kappa_yh_{y}-\kappa_z h_{z},\\
\nonumber h_{z}(0)&=0.
\end{align}
Solving the system, we get
\begin{align}
\nonumber h_{0}(\tau)&=\frac{\theta_z}{\kappa_z-\kappa_y}(\kappa_z(1-e^{-\kappa_y t})-\kappa_y(1-e^{-\kappa_z \tau}))\\
\nonumber h_{y}(\tau)&=e^{-\kappa_y \tau}, \quad h_{z}(\tau)=e^{-\kappa_z \tau}\frac{\kappa_y}{\kappa_z-\kappa_y}(e^{\tau(\kappa_z-\kappa_y)}-1)
\end{align}
implying that \begin{align}\label{Ey}
\begin{split}h(\tau,y,z)&=\frac{\theta_z}{\kappa_z-\kappa_y}(\kappa_z(1-e^{-\kappa_y \tau})-\kappa_y(1-e^{-\kappa_z \tau}))+e^{-\kappa_y \tau}y\\
&+e^{-\kappa_z \tau}\frac{\kappa_y}{\kappa_z-\kappa_y}(e^{\tau(\kappa_z-\kappa_y)}-1)z
\end{split}\end{align}
{\textit(iii)} Let $E[Y_{t+\tau}Z_{t+\tau}|Y_t=y,Z_t=z]=:f(\tau,Y_t,Z_t)$. $f$ solves formally 
\begin{align}\label{f}
\partial_{\tau}f&=\kappa_y(z-y)\partial_y f+\kappa_z(\theta_z-z)\partial_z f+\displaystyle\frac{1}{2}\sigma_y^2y\partial_{yy}f+\displaystyle\frac{1}{2}\sigma_z^2z\partial_{zz} f.
\end{align}
Following exactly the same lines as above we have
\begin{equation}\label{momentf2}
f(\tau,y,z)=f_{0}(\tau)+f_{y}(\tau)y+f_{z}(\tau)z+f_{z^2}(\tau)z^2+f_{zy}(\tau)zy+f_{y^2}(\tau)y^2.
\end{equation}
Plugging \eqref{momentf2} into \eqref{f} gives
\begin{align}
\nonumber \begin{split}\displaystyle\frac{d}{d\tau}{f}_{0}+\displaystyle\frac{d}{d\tau}{f}_{y}y+\displaystyle\frac{d}{d\tau}{f}_{z}z+\displaystyle\frac{d}{d\tau}{f}_{z^2}z^2+\displaystyle\frac{d}{d\tau}{f}_{zy}zy+\displaystyle\frac{d}{d\tau}{f}_{y^2}y^2&=  \kappa_y(z-y)(f_y+2f_{y^2}y+f_{zy}z)\\
&+\kappa_z(\theta_z-z)(f_z+f_{zy}y+2f_{z^2}z)\\
&+\sigma_y^2yf_{y^2}+\sigma_z^2zf_{z^2}\end{split}.
\end{align}
Thus we have
\begin{align}
\nonumber\displaystyle\frac{d}{d\tau}{f}_{0}&=\kappa_z\theta_zf_{z},\\
\nonumber\displaystyle\frac{d}{d\tau}{f}_{y}&=-\kappa_yf_{y}+\kappa_z\theta_zf_{zy}+\sigma_y^2f_{y^2},\\
\nonumber\displaystyle\frac{d}{d\tau}{f}_{z}&=\kappa_yf_{y}-\kappa_z f_{z}+(2\kappa_z\theta_z+\sigma_z^2)f_{z^2},\\
\nonumber\displaystyle\frac{d}{d\tau}{f}_{z^2}&=\kappa_yf_{zy}-2\kappa_zf_{z^2},\\
\nonumber\displaystyle\frac{d}{d\tau}{f}_{y^2}&=-2\kappa_yf_{y^2},\\
\nonumber\displaystyle\frac{d}{d\tau}{f}_{zy}&=2\kappa_yf_{y^2}-(\kappa_y+\kappa_z)f_{zy},
\end{align}
with \[f_{0}(0)=f_{z}(0)=f_{y}(0)=f_{z^2}(0)=f_{y^2}(0)=0,\quad  f_{zy}(0)=1.\]
Solving the above system yields
\begin{align}
\nonumber \begin{split} f_0(\tau)&=\frac{e^{-(2\kappa_z+\kappa_y)\tau}\theta_z}{2(\kappa_z\kappa_y^2-\kappa_z^3)}\Big(2e^{2\kappa_z\tau}(\kappa_z+\kappa_y)\kappa_z^2\theta_z+e^{\kappa_yt}(\kappa_z+\kappa_y)\kappa_y(\sigma_z^2+2\kappa_z\theta_z)\\
&-2e^{\kappa_z\tau}\kappa_z^2(\sigma_z^2+(\kappa_z+\kappa_y)\theta_z)-e^{(2\kappa_z+\kappa_y)\tau}(\kappa_z-\kappa_y)(2\kappa_z^2\theta_z+\kappa_y(\sigma_z^2+2\kappa_z\theta_z))\\
&+2e^{(\kappa_z+\kappa_y)\tau}(\kappa_z+\kappa_y)(-\kappa_y\sigma_z^2+\kappa_z^2\theta_z+\kappa_z(\sigma_z^2-2\kappa_y\theta_z)) \Big),\end{split}\\
\nonumber f_y(\tau)&=\theta_z(e^{-\kappa_y \tau}-e^{-(\kappa_z+\kappa_y)\tau}),\\
\nonumber \begin{split}f_z(\tau)&=\frac{e^{-(2\kappa_z+\kappa_y)\tau}}{\kappa_z(\kappa_z-\kappa_y)}\Big(e^{2\kappa_z\tau}\kappa_z\kappa_y\theta_z+e^{\kappa_y\tau}\kappa_y(\sigma_z^2+2\kappa_z\theta_z)-e^{\kappa_z\tau}\kappa_z(\sigma_z^2+(\kappa_z+\kappa_y)\theta_z)\\
&+e^{(\kappa_z+\kappa_y)\tau}(-\kappa_y\sigma_z^2+\kappa_z^2\theta_z+\kappa_z(\sigma_z^2-2\kappa_y\theta_z))\Big),\end{split}\\
\nonumber f_{y^2}&\equiv0, \ f_{zy}=e^{-(\kappa_z+\kappa_y)\tau}, \ f_{z^2}(\tau)=\frac{\kappa_y}{\kappa_z-\kappa_y}(e^{-(\kappa_z+\kappa_y)\tau}-e^{-2\kappa_z \tau}).
\end{align}
Inserting these expressions into \eqref{momentf2} we get $f(\tau,y,z)$.
\par{\emph(iv)} Set $E[Z_{t+\tau}^2|Y_t=y,Z_t=z]=:g(t,Y_t,Z_t)$. Then, $g$ solves
\begin{align}\label{q}
\partial_{\tau}g=\kappa_y(z-y)\partial_y g+\kappa_z(\theta_z-z)\partial_z g+\displaystyle\frac{1}{2}\sigma_y^2y\partial_{yy} g+\displaystyle \frac{1}{2}\sigma_z^2z\partial_{zz} g.
\end{align}
Also, $g$ is in the following form:
\begin{align}\label{momentq2}
g(\tau,y,z)=q_{0}(t)+q_{y}(t)y+q_{z}(t)z+q_{z^2}(t)z^2+q_{zy}(t)zy+q_{y^2}(t)y^2.
\end{align}
Inserting \eqref{momentq2} into \eqref{q} gives
\begin{align}\begin{split}
\nonumber\displaystyle\frac{d}{d\tau}{q}_{0}+\displaystyle\frac{d}{d\tau}{q}_{y}y+\displaystyle\frac{d}{d\tau}{q}_{z}z+\displaystyle\frac{d}{d\tau}{q}_{z^2}z^2+\displaystyle\frac{d}{d\tau}{q}_{zy}zy+\displaystyle\frac{d}{d\tau}{q}_{y^2}y^2&=\kappa_y(z-y)(q_y+2q_{y^2}y+q_{zy}z)\\
&+\kappa_z(\theta_z-z)(q_z+q_{zy}y+2q_{z^2}z)\\
&+\sigma_y^2yq_{y^2}+\sigma_z^2zq_{z^2}
\end{split}.\end{align}
which yields the system
\begin{eqnarray}
\nonumber\displaystyle\frac{d}{d\tau}{q}_{0}&=&\kappa_z\theta_zq_{z},\\
\nonumber\displaystyle\frac{d}{d\tau}{q}_{y}&=&-\kappa_yq_{y}+\kappa_z\theta_zq_{zy}+\sigma_y^2q_{y^2},\\
\nonumber\displaystyle\frac{d}{d\tau}{q}_{z}&=&\kappa_yq_{y}-\kappa_z q_{z}+(2\kappa_z\theta_z+\sigma_z^2)q_{z^2},\\
\nonumber\displaystyle\frac{d}{d\tau}{q}_{z^2}&=&\kappa_yq_{zy}-2\kappa_zq_{z^2},\\
\nonumber\displaystyle\frac{d}{d\tau}{q}_{y^2}&=&-2\kappa_yq_{y^2},\\
\nonumber\displaystyle\frac{d}{d\tau}{q}_{zy}&=&2\kappa_yq_{y^2}-(\kappa_y+\kappa_z)q_{zy},
\end{eqnarray}
with \begin{align}\nonumber q_{0}(0)=q_{z}(0)=q_{y}(0)=q_{y^2}(0)=q_{zy}(0)=0,\quad  q_{z^2}(0)=.1\end{align}
We solve this system of equations and get
\begin{align}
\nonumber q_0(\tau)&=\displaystyle \frac{e^{-2\kappa_z\tau}(e^{\kappa_z\tau}-1)^2\theta_z(\sigma_z^2+2\kappa_z\theta_z)}{2\kappa_z},\\
\nonumber \\
\nonumber q_z(\tau)&=\displaystyle\frac{e^{-2\kappa_z\tau}(e^{\kappa_z\tau}-1)(\sigma_z^2+2\kappa_z\theta_z)}{\kappa_z},\\
\nonumber \\
\nonumber q_{z^2}(\tau)&=e^{-2\kappa_z \tau},\\
\nonumber \\
\nonumber q_{y}(\tau)&\equiv q_{y^2}(\tau)\equiv q_{zy}(\tau) \equiv0.
\end{align}
Inserting above expressions into \eqref{momentq2} yields the expression for $g$.

\par{\emph(v)} Let $E[Y_{t+\tau}^2|Y_t=y,Z_t=z]=:h(\tau,Y_t,Z_t)$. Formally, $h$ satisfies the Kolmogorov's backward equation
\begin{align}\label{p}
\partial_{\tau} h&=\kappa_y(z-y)\partial_y h+\kappa_z(\theta_z-z)\partial_z h+\displaystyle\frac{1}{2}\sigma_y^2y\partial_{yy} h+\displaystyle\frac{1}{2}\sigma_z^2z\partial_{zz} h.
\end{align}
From the polynomial property $h$ is of the form
\begin{align}\label{momentp2}
h(\tau,y,z)=p_{0}(\tau)+p_{y}(\tau)y+p_{z}(\tau)z+p_{z^2}(\tau)z^2+p_{zy}(\tau)zy+p_{y^2}(\tau)y^2.
\end{align}
Plugging \eqref{momentp2} into \eqref{p} gives
\begin{align}
\nonumber \begin{split}\displaystyle\frac{d}{d\tau}{p}_{0}+\displaystyle\frac{d}{d\tau}{p}_{y}y+\displaystyle\frac{d}{d\tau}{p}_{z}z+\displaystyle\frac{d}{d\tau}{p}_{z^2}z^2+\displaystyle\frac{d}{d\tau}{p}_{zy}zy+\displaystyle\frac{d}{d\tau}{p}_{y^2}y^2&=\kappa_y(z-y)(p_y+2p_{y^2}y+p_{zy}z)\\
&+\kappa_z(\theta_z-z)(p_z+p_{zy}y+2p_{z^2}z)\\
&+\sigma_y^2yp_{y^2}+\sigma_z^2zp_{z^2}\end{split}.
\end{align}
which yields the following system of differential equations
\begin{align}
\nonumber\displaystyle\frac{d}{d\tau}{p}_{0}&=\kappa_z\theta_zp_{z},\\
\nonumber\displaystyle\frac{d}{d\tau}{p}_{y}&=-\kappa_yp_{y}+\kappa_z\theta_zp_{zy}+\sigma_y^2p_{y^2},\\
\nonumber\displaystyle\frac{d}{d\tau}{p}_{z}&=\kappa_yp_{y}-\kappa_z p_{z}+(2\kappa_z\theta_z+\sigma_z^2)p_{z^2},\\
\nonumber\displaystyle\frac{d}{d\tau}{p}_{z^2}&=\kappa_yp_{zy}-2\kappa_zp_{z^2},\\
\nonumber\displaystyle\frac{d}{d\tau}{p}_{y^2}&=-2\kappa_yp_{y^2},\\
\nonumber\displaystyle\frac{d}{d\tau}{p}_{zy}&=2\kappa_yp_{y^2}-(\kappa_y+\kappa_z)p_{zy},
\end{align}
with
\begin{align}
\nonumber p_{0}(0)=p_{z}(0)=p_{y}(0)=p_{z^2}(0)=p_{zy}(0)=0,\quad  p_{y^2}(0)=1.
\end{align}

Solving the system of linear ODEs yields
\begin{align}\nonumber
\begin{split}
p_0(\tau)&=\displaystyle\frac{e^{-(3\kappa_z+\kappa_y)\tau}\theta_z}{2\kappa_z(\kappa_z-2\kappa_y)(\kappa_z-\kappa_y)^2\kappa_y(\kappa_z+\kappa_y)}
\Big(e^{(\kappa_z+\kappa_y)\tau}(\kappa_z-2\kappa_y)\kappa_y^3(\kappa_z+\kappa_y)\\
&\times(\sigma_z^2+2\kappa_z\theta_z)-2e^{3\kappa_z\tau}\kappa_z^2(\kappa_z-2\kappa_y)(\kappa_z^2-\kappa_y^2)(\sigma_y^2+2\kappa_y\theta_z)
-4e^{2\kappa_z\tau}\kappa_z^2\\
&\times(\kappa_z-2\kappa_y)\kappa_y^2(\sigma_z^2+(\kappa_z+\kappa_y)\theta_z)+e^{(3\kappa_z+\kappa_y)\tau}(\kappa_z-2\kappa_y)(\kappa_z-\kappa_y)^2\\
&\times(\kappa_z^2\sigma_y^2+\kappa_z\kappa_y\sigma_y^2+\kappa_y^2\sigma_z^2+2\kappa_z\kappa_y(\kappa_z+\kappa_y)\theta_z)+e^{(3\kappa_z-\kappa_y)\tau}\kappa_z^2(\kappa_z+\kappa_y)\\
&\times(\kappa_y^2(\sigma_y^2-\sigma_z^2)+\kappa_z^2(\sigma_y^2+2\kappa_y\theta_z)-2\kappa_z\kappa_y(\sigma_y^2+2\kappa_y\theta_z))+2e^{(2\kappa_z+\kappa_y)\tau}\kappa_y^2\\
&\times(\kappa_y^2-\kappa_z^2)(2\kappa_y\sigma_z^2-2\kappa_z^2\theta_z+\kappa_z(\sigma_y^2-2\sigma_z^2+4\kappa_y\theta_z))\Big) ,\end{split}\end{align}
\begin{align}
\nonumber p_{y}(\tau)&=\displaystyle\frac{e^{-2\kappa_y\tau}((1-e^{\kappa_y\tau})\kappa_z(\sigma_y^2+2\kappa_y\theta_z)+\kappa_y((e^{\kappa_y\tau}-1)\sigma_y^2+2(e^{\kappa_y\tau}-e^{(\kappa_y-\kappa_z)\tau})\kappa_y\theta_z))}{\kappa_y(\kappa_y-\kappa_z)},\end{align}
\begin{align}
\nonumber \begin{split}
p_{z}(\tau)&=\displaystyle\frac{e^{-(3\kappa_z+\kappa_y)\tau}}{\kappa_z(\kappa_z-2\kappa_y)(\kappa_z-\kappa_y)^2}\Big(-e^{(\kappa_z+\kappa_y)\tau}(\kappa_z-2\kappa_y)\kappa_y^2(\sigma_z^2+2\kappa_z\theta_z)\\
&+e^{3\kappa_z\tau}\kappa_z(\kappa_z-2\kappa_y)(\kappa_z-\kappa_y)(\sigma_y^2+2\kappa_y\theta_z)+2e^{2\kappa_z\tau}\kappa_(\kappa_z-2\kappa_y)\kappa_y\\
&\times(\sigma_z^2+(\kappa_z+\kappa_y)\theta_z)-e^{(3\kappa_z-\kappa_y)\tau}\kappa_z(\kappa_y^2(\sigma_y^2-\sigma_z^2)+\kappa_z^2(\sigma_y^2+2\kappa_y\theta_z)\\
&-2\kappa_z\kappa_y(\sigma_y^2+2\kappa_y\theta_z))-e^{(2\kappa_z+\kappa_y)\tau}\kappa_y(\kappa_y-\kappa_z)(2\kappa_y\sigma_z^2-2\kappa_z^2\theta_z\\
&+\kappa_z(\sigma_y^2-2\sigma_z^2+4\kappa_y\theta_z))\Big),
\end{split}
\end{align}
\begin{align}
\nonumber p_{zy}(\tau)&=\displaystyle\frac{2\kappa_y e^{-(\kappa_z+\kappa_y)\tau}(e^{(\kappa_z-\kappa_y)\tau}-1)}{\kappa_z-\kappa_y},\\
\nonumber p_{y^2}(\tau)&=e^{-2\kappa_y \tau},\\
\nonumber p_{z^2}(\tau)&=\displaystyle \frac{\kappa_y^2e^{-2\kappa_z\tau}(e^{(\kappa_z-\kappa_y)\tau}-1)^2}{(\kappa_z-\kappa_y)^2}.
\end{align}
Finally, inserting these coefficients into \eqref{momentp2} gives $h(\tau,y,z)$.

Now we need to prove that the polynomial expressions in $(i)-(v)$ actually solve \eqref{infgen1}, that is, they provide the conditional moments of $X_{t+\tau}$. The next lemma gives a criteria for this to hold. Clearly, functions $f$, $g$ and $h$ appearing in $(i)-(v)$ above are $C^{1,2}$ functions whose spatial derivatives satisfy the polynomial growth condition given in
\eqref{biktim}, meaning that the result of Lemma \ref{FLL} applies and this finishes the proof of the Proposition~\ref{condmoment}.
\end{proof}
\begin{lem}\label{FLL} Suppose $u_0$ is a $C^2$-function on $\mathcal{X}$, and  $u$ is a $C^{1,2}$ -function on $R^+\times \mathcal{X}$ whose spatial derivatives satisfy the polynomial growth condition
\begin{align} \label{biktim}
\left\|\left(\frac{\partial u}{\partial y} \ , \ \frac{\partial u}{\partial z} \right)  \right\|\leq K(1+\|(y,z) \|^{\nu}), \ t\leq T, \quad (y,z) \in \mathcal{X}
\end{align}
for some constant $K=K(T)\leq \infty$ and some ${\nu}\geq 1$, for all $T<\infty$.

If $u(t,y,z)$ satisfies the Kolmogorov backward equation
\begin{align}\label{btsnn}
\nonumber \frac{\partial u}{\partial t}&=\kappa_y(z-y)\frac{\partial u}{\partial y}+\kappa_z(\theta_z-z)\frac{\partial u}{\partial z} +\frac{1}{2}\sigma_y^2y\frac{\partial^2 u}{ \partial y^2} +\frac{1}{2}\sigma_z^2z\frac{\partial^2 u}{\partial z^2},\\
u(0,y,z)&=u_0(y,z)
\end{align}
for all $t\geq 0$ and $(y,z)\in \mathcal{X}$, then for all $t\leq T<\infty$
\[
u(T-t,Y_t,Z_t)=E[u_0(Y_T,Z_T)|Y_t,Z_t]\]
\end{lem}
\begin{proof}
Since $u$ is assumed to be $C^{1,2}$, in view of the It\^{o} formula we get
\begin{align}\label{btsnn2}
\begin{split}d u(T-t,Y_t,Z_t)&=\Bigg( -\frac{\partial u(T-t,Y_t,Z_t)}{\partial t}+\frac{\partial u(T-t,Y_t,Z_t)}{\partial y}\kappa_y(Z_t-Y_t)\\
&+\frac{\partial u(T-t,Y_t,Z_t)}{\partial z} \kappa_z(\theta_z-Z_t)+\frac{1}{2}\sigma_y^2Y_t\frac{\partial^2 u(T-t,Y_t,Z_t)}{ \partial y^2}\\
& +\frac{1}{2}\sigma_z^2Z_t\frac{\partial^2 u(T-t,Y_t,Z_t)}{\partial z^2} \Bigg) dt\\
&+ \frac{\partial u(T-t,Y_t,Z_t)}{\partial y}\sigma_y\sqrt{Y_t}dW^y_t+\frac{\partial u(T-t,Y_t,Z_t)}{\partial z}\sigma_z\sqrt{Z_t}dW^z_t
\end{split}\end{align}
Now suppose $u$ satisfies \eqref{btsnn}. Then, the drift term in \eqref{btsnn2} immediately vanishes, implying that $u(T-t,Y_t,Z_t)$ is a local martingale with $u(0,Y_T,Z_T)=u_0(Y_T,Z_T)$. We now write
\[d u(T-t,Y_t,Z_t)= \frac{\partial u(T-t,Y_t,Z_t)}{\partial y}\sigma_y\sqrt{Y_t}dW^y_t+\frac{\partial u(T-t,Y_t,Z_t)}{\partial z}\sigma_z\sqrt{Z_t}dW^z_t
\]
In what follows our main objective is to show that under the assumptions of the lemma, $u(T-t,Y_t,Z_t)$ is a true martingale.  

We have
\begin{align}
\nonumber
&E\left[\int_0^T \left\|   \left(\frac{\partial u(T-s,Y_s,Z_s)}{\partial y}, \ \frac{\partial u(T-s,Y_s,Z_s)}{\partial z}\right)\left[\begin{array}{cc} \sigma_y \sqrt{Y_t}&0\\0&\sigma_z\sqrt{Z_t}\end{array}\right]\right\|^2 ds\right] \\
\nonumber \leq & E\left[\int_0^T \left\|   \left(\frac{\partial u(T-s,Y_s,Z_s)}{\partial y}, \ \frac{\partial u(T-s,Y_s,Z_s)}{\partial z}\right)\right\|^2  \left\| \left[\begin{array}{cc} \sigma_y^2 Y_t&0\\0&\sigma_z^2 Z_t\end{array}\right]   \right\| ds  \right] \\
\label {bittik}\leq & K\left(1+E\left[ \sup_{s\leq T} \left\| (Y_s,Z_s)\right\|^{2\nu}   \right]\right)
\end{align}
where the last inequality follows from assumption \eqref{biktim} and due to the fact that the diffusion parameter of the process $X$ satisfies the linear growth condition. Finally, one can show that (see, for example, \citet[Problem 5.3.15]{kar}) the expectation in \eqref{bittik} is finite and this yields the desired result.
\end{proof}

\end{appendix}

\bibliography{referencessum}
\bibliographystyle{apa}

\end{document}